\documentclass[sn-mathphys,Numbered]{sn-jnl}

\usepackage{manyfoot}
\usepackage{proof}
\usepackage{bm}
\usepackage{latexsym}
\usepackage{amssymb}
\usepackage{amsmath}
\usepackage{amsfonts}
\usepackage{amsthm}
\usepackage{algorithm}
\usepackage[noend]{algorithmic}
\usepackage{verbatim}
\usepackage{color}
\usepackage{graphicx}

\newcommand{\ifdraft}[1]{#1}
\definecolor{aocolour}{rgb}{0.7,0.8,1}
\newcommand{\ao}[1]{\ifdraft{\noindent\colorbox{aocolour}{A.O.: #1}}}

\newcommand{\PC}{\mathbf{Pr}}
\newcommand{\Conj}{\mathbf{Conj}}
\newcommand{\BCat}{\mathbf{BCat}}
\newcommand{\Cat}{\mathbf{BCat}_{\wedge}}
\newcommand{\Tp}{\mathbf{Cat}}
\newcommand{\C}{\mathop{\raisebox{.2pt}{\&}}}
\newcommand{\SL}{\mathop{/}}
\newcommand{\BS}{\mathop{\backslash}}

\newcommand{\LL}{\mathbf{L}}
\newcommand{\Ls}{\mathbf{L}\!^{\boldsymbol{*}}}
\newcommand{\MALC}{\mathbf{MALC}}
\newcommand{\MALCs}{\mathbf{MALC}^{\boldsymbol{*}}}
\newcommand{\MACLL}{\mathbf{MACLL}}
\newcommand{\MALCD}{\mathbf{MALCD}}
\newcommand{\MALCDs}{\mathbf{MALCD}^{\boldsymbol{*}}}
\newcommand{\conjcat}{\mathcal{C}}

\newcommand{\Par}{\mathop{\raisebox{.6em}{\rotatebox{180}{\&}}}}

\newcommand{\mathsc}[1]{{\normalfont\textsc{#1}}}

\renewcommand{\epsilon}{\varepsilon}

\theoremstyle{thmstyleone}%
\newtheorem{theorem}{Theorem}
\newtheorem{lemma}{Lemma}
\newtheorem{oldtheorem}{Theorem}

\newtheorem{oldlemma}[oldtheorem]{Lemma}
\newtheorem{corollary}{Corollary}

\theoremstyle{thmstyletwo}%
\newtheorem{example}{Example}%
\newtheorem*{claim}{Claim}

\theoremstyle{thmstylethree}%

\begin{document}

\title{Conjunctive categorial grammars and Lambek grammars with additives}


\author*[1,2]{\fnm{Stepan L.} \sur{Kuznetsov}}\email{sk@mi-ras.ru}

\author[3]{\fnm{Alexander} \sur{Okhotin}}\email{alexander.okhotin@spbu.ru}

\affil[1]{\orgname{Steklov Mathematical Institute of RAS},
\orgaddress{\street{Gubkina St., 8}, \city{Moscow}, \postcode{119991}, \country{Russia}}}

\affil[2]{\orgname{Department of Computer Science, HSE University},
\orgaddress{\street{Pokrovsky Blvd., 11},
\city{Moscow}, \postcode{109028}, \country{Russia}}}

\affil[3]{\orgdiv{Department of Mathematics and Computer Science}, \orgname{St.~Petersburg State University},
\orgaddress{\street{14th Line V. O., 29}, \city{Saint Petersburg}, \postcode{199178}, \country{Russia}}}


\abstract{%
A new family of categorial grammars is proposed,
defined by enriching basic categorial grammars with a conjunction operation.
It is proved that the formalism obtained in this way
has the same expressive power as conjunctive grammars,
that is, context-free grammars enhanced with conjunction.
It is also shown
that categorial grammars with conjunction
can be naturally embedded
into the Lambek calculus with conjunction and disjunction operations.
This further implies that a certain NP-complete set
can be defined in the Lambek calculus with conjunction.
We also show how to handle some subtle issues connected with the empty string.
Finally, we prove that a language generated by a conjunctive grammar can be
described by a Lambek grammar with disjunction (but without conjunction).
\footnotetext{This article is an extended version of the conference presentation ``Conjunctive categorial grammars'' at the Mathematics of Language
2017 meeting (London, UK, July 13--14, 2017; proceedings published in ACL Anthology, \href{http://aclweb.org/anthology/W17-3414}{W17-3414}).}
}




\maketitle

\sloppy

\section{Introduction}

This paper establishes a connection
between two formal grammar models
that emerged within two different theories of syntax.

One theory is the \emph{immediate constituent analysis},
which has its roots in the traditional grammar,
and was investigated in the early 20th century
by the structural linguists.
It reached universal recognition
under the name ``context-free grammars,''
introduced in Chomsky's early work.
In this paradigm,
a grammar assigns certain properties
to groups of words,
such as ``noun phrase'' (NP), ``verb phrase'' (VP) or ``sentence'' (S).
These properties are known as \emph{syntactic categories},
or, in Chomsky's terminology, \emph{nonterminal symbols}.
Rules of a grammar, such as $\mathrm{S} \to \mathrm{NP}\ \mathrm{VP}$,
show how shorter substrings with known properties
can be concatenated to form longer strings belonging to a certain category.

In the other theory,
which was discovered by Ajdukiewicz~\cite{Ajdukiewicz1935} and by Bar-Hillel~\cite{Barhillel_categorial},
and is nowadays known as \emph{categorial grammars},
syntactic categories are defined in a different way.
``Noun phrases'' are treated as ``the category of phrases \emph{equivalent} to nouns'',
that is, admissible wherever a noun would be admissible.
Then, verb phrases are defined as
``the category of phrases, which would form a complete sentence,
if a noun, or anything equivalent to a noun, is concatenated on the left'',
denoted by $\mathsc{noun} \backslash \mathsc{sentence}$.
A categorial grammar explicitly assigns categories to individual words;
and then, by definition,
a concatenation of any string of type $\mathsc{noun}$
with any string of type $\mathsc{noun} \backslash \mathsc{sentence}$
forms a complete sentence.
The crucial feature of this approach
is that the laws that govern reduction of categories,
namely, $A \, (A \BS B)$ to $B$,
and $(B \SL A) \, A$ to $B$,
are universal.
In contrast, Chomsky's context-free formalism
uses different rules for different categories.

A formal connection between these two models
was established by Bar-Hillel et al.~\cite{BGS1960},
who proved them to be equivalent in power:
a language is defined by a context-free grammar
if and only if
it is defined by a basic categorial grammar
(assuming languages without the empty string).

More than half a century of research
gave birth to many extensions of both basic models.
Categorial grammars were the first to get an interesting extension:
the \emph{Lambek calculus}, introduced by Lambek~\cite{Lambek1958},
augments the model with additional derivation rules.
Later, Pentus~\cite{Pentus1993} established
that this extended model is still equivalent in power
to context-free grammars.
Pentus's translation yields a context-free grammar
of exponential size with respect to the original Lambek grammar.
For the special case of {\em unidirectional} Lambek grammars,
which use only one kind of division operators ($\BS$, $\SL$), but not both,
Kuznetsov~\cite{Kuznetsov2016},
using the ideas of Savateev~\cite{Savateev2009},
presents a polynomial translation into context-free grammars.
Other generalizations of categorial grammars include
combinatory categorial grammars by Steedman~\cite{Steedman},
categorial dependency grammars by Dikovsky and Dekhtyar~\cite{DikovskyDekhtyar},
and others.

Lambek grammars, in their turn, can be generalized further.
Modern extensions of the Lambek calculus with new operations,
such as those considered
by Carpenter~\cite{CarpenterBook},
Morrill~\cite{MorrillBook}
and Moot and Retor\'{e}~\cite{MootRetoreBook},
are capable of describing quite sophisticated syntactic phenomena.

From the point of view of modern logic, the Lambek calculus is a
substructural logical system, namely, a non-commutative variant
of \emph{linear logic},
introduced by Girard~\cite{Girard1987},
see Abrusci~\cite{Abrusci}, Yetter~\cite{Yetter}.
Linear logic offers many logical operations,
and some of them can be used
in the non-commutative case
for extending Lambek grammars.

Morrill~\cite{MorrillBook} and his collaborators, following and extending
Moortgat~\cite{Moortgat1996}, consider a system, based on the Lambek calculus,
with discontinuous connectives, subexponentials for controlled non-linearity,
brackets for controlled non-associativity, and many other operations.
The use of negation in categorial grammars
was considered by Buszkowski~\cite{Buszkowski_negation}.
Kanazawa~\cite{Kanazawa1992} investigated 
the power of Lambek grammars with conjunction and disjunction operations
that are ``additive operations'' in terms of linear logic.

Numerous generalized models have also been introduced
in the paradigm of immediate constituent analysis,
as extensions of the context-free formalism.
One direction is to extend the form of constituents,
that is, sentence fragments
to which syntactic categories are being assigned in a grammar.
The most well-known of these models are the multi-component grammars,
introduced by Vijay-Shanker et al.~\cite{VijayshankerWeirJoshi}
and by Seki et al.~\cite{SekiMatsumuraFujiiKasami},
inspired by an earlier model by Fischer~\cite{Fischer}:
these grammars define the properties of discontinuous constituents,
that is, substrings with a bounded number of gaps.
Extensions of another kind augment the model
by introducing new logical operators to be used in grammar rules:
for instance, \emph{conjunctive grammars},
featuring a conjunction operation,
and \emph{Boolean grammars},
further equipped with negation,
were introduced by Okhotin~\cite{Conjunctive,BooleanGrammars}.
Earlier, Latta and Wall~\cite{LattaWall}
argued for the relevance of such operations
in linguistic descriptions.
The main results on conjunctive grammars indicate that they preserve
the practically useful properties of context-free grammars,
such as efficient parsing algorithms,
while substantially extending their expressive power.
The known results on conjunctive grammars
are presented in a recent survey by Okhotin~\cite{ConjunctiveTokyo}.

A decade ago, Kuznetsov~\cite{Kuznetsov2013}
compared the expressive power of Lambek grammars
with conjunction,
as considered by Kanazawa~\cite{Kanazawa1992},
with that of conjunctive grammars.
It was proved that a large subclass of conjunctive grammars
(namely, conjunctive grammars in Greibach normal form)
can be simulated in the Lambek calculus with conjunction,
but the exact power of the latter remains undetermined.

This paper makes a fresh attempt
at introducing conjunction in categorial grammars.
The new model extends basic categorial grammars,
rather than Lambek grammars,
and for that reason it uses categories and rules
of a simpler form than in the earlier model
of Kanazawa~\cite{Kanazawa1992}
and Kuznetsov~\cite{Kuznetsov2013}.
Yet, it is shown that this model can simulate every conjunctive grammar.
A converse simulation is presented as well,
which implies the equivalence of the two models.

The grammars studied in this paper
are modelled on the classical notions of a basic categorial grammar
and a context-free grammar, presented in the introductory Section~\ref{S:intro}.
In this paper, context-free grammars are defined via logical derivation
rather than Chomskian string rewriting,
and this exposes further similarities between the two models.
Then, Section~\ref{section_definitions_conjunction} similarly presents conjunctive grammars,
defined via logical derivation as well,
and introduces conjunctive categorial grammars.

The equivalence between these two models is
established in Section~\ref{section_equivalence}.
As compared to the classical equivalence result
for context-free grammars and basic categorial grammars,
the new result requires a more elaborate construction.
One particular difficulty
is that the normal form theorems for conjunctive grammars
are weaker than those for the context-free grammars:
in particular, no analogue of the Greibach normal form is
known for conjunctive grammars.
For this reason, the simulation of conjunctive grammars
by the proposed conjunctive categorial grammars
relies on a different normal form by Okhotin and Reitwie{\ss}ner~\cite{OkhotinReitwiessner}.
This leads to a representation
of the whole class of conjunctive grammars,
in contrast to the result by Kuznetsov~\cite{Kuznetsov2013},
which is valid only for grammars in Greibach normal form.

The second result, established in Section~\ref{section_lambek},
is that conjunctive categorial grammars,
as defined in this paper,
can be represented in the Lambek calculus
with the conjunction operation,
as considered by Kanazawa~\cite{Kanazawa1992},
and therefore this extension of the Lambek calculus
is at least as powerful as are the conjunctive grammars.
Furthermore, it is proved that Kanazawa's~\cite{Kanazawa1992} model
can describe an NP-complete language,
which conjunctive grammars cannot describe unless $\mathrm{P}=\mathrm{NP}$.

In Section~\ref{S:empty}, we discuss issues connected with the empty string. Namely,
the Lambek calculus has a natural variant which allows grammars to describe languages
which include the empty string, and we show that any language described by a conjunctive
grammar, even if it includes the empty string, can be also described by a grammar based
on this variation of the Lambek calculus with conjunction.

Finally, in Section~\ref{S:disjunction} we show that our results on embedding
conjunctive grammars can be equivalently established for the Lambek calculus extended with
disjunction instead of conjunction. In Section~\ref{S:distributivity} we prove that the
results keep valid if one adds the distributivity principle for conjunction and disjunction.

\section{Basic Categorial Grammars and Context-Free Grammars}\label{S:intro}

Let $\Sigma$ be a finite \emph{alphabet} of the language being defined,
its elements are called \emph{symbols}.
In linguistic descriptions,
symbols typically represent words of the language.
The set of all strings over $\Sigma$, including the empty one ($\varepsilon$), is denoted by $\Sigma^*$.
The notation $\Sigma^+$ stands for the set of all non-empty strings.
Any subset of $\Sigma^*$ is a (formal) \emph{language.} We distinguish languages with and without
the empty string; the latter are subsets of $\Sigma^+$.

The models considered in this paper
are derived from two classical formal grammar frameworks:
basic categorial grammars
and context-free grammars.

{\em Basic categorial grammars} (BCG) have their roots in the works of Ajdukiewicz~\cite{Ajdukiewicz1935}.
Let $\Sigma$ be an alphabet.
Let $\PC = \{ p, q, r, \dots \}$ be a finite set of \emph{primitive categories},
and let $s \in \PC$ be a designated \emph{target category}
of all syntactically correct sentences.

The set $\BCat$ of {\em basic categories} is defined as the smallest set which satisfies the following conditions.
\begin{enumerate}
\item
	Every primitive category is a basic category.
\item
	If $A \in \BCat$ and $p \in \PC$,
	then $(p \BS A) \in \BCat$ and $(A \SL p) \in \BCat$.
\end{enumerate}

The definition of a basic categorial grammar is given
in terms of logical \emph{propositions},
which are expressions of the form $B(v)$,
where $B \in \BCat$ and $v \in \Sigma^+$.
This proposition states that $v$ is a string of syntactic category $B$.

A categorial grammar is then regarded
as a \emph{logical calculus}
for deriving categorial propositions.
It includes a finite set of axioms (axiomatic propositions) of the form $A(a)$,
where $A \in \BCat$ and $a \in \Sigma$,
and the following inference rules.
\begin{equation*}
	\infer{A(uv)}{p(u) & (p \BS A)(v)}
	\qquad
	\infer{A(uv)}{(A \SL p)(u) & p(v)}
\end{equation*}

The string $w$ belongs to the language generated by the BCG
if and only if
the proposition $s(w)$ is derivable in this calculus.

By definition, BCGs can describe only languages without the empty string.

\begin{example}\label{BCG_example}
The basic categorial grammar
with the following axiomatic propositions,
where $\PC = \{ s, p, q \}$
and $s$ is the target category,
describes the language $\{b a^n c a^n \mid n \geqslant 0 \}$.
\begin{equation*}
	(s \SL p)(b), \quad p(c), \quad (p \SL q)(a), \quad (p \BS q)(a)
\end{equation*}
\end{example}

Another, universally known formal grammar framework
is the phrase-structure formalism,
defined by Chomsky~\cite{Chomsky}
and later renamed into \emph{context-free grammars} (CFG).
In a CFG,
there is a fixed finite set of categories $N$
(usually called ``non-terminal symbols''),
and one of them is designated as the initial symbol $S \in N$.
The grammar is defined by a finite set of rules
(or ``productions'')
of the form $A \to \beta$,
where $A \in N$ and $\beta \in (\Sigma \cup N)^*$.

It is a well-known fact that if a language is context-free and does not contain $\varepsilon$, then it can be described by a CFG
without $\varepsilon$-rules: for each rule $A \to \beta$, $\beta \in (\Sigma \cup N)^+$.

Even though Chomsky's original definition of context-free grammars
was given in terms of string rewriting,
it is more convenient---at least in this paper---%
to present it as a logical derivation
similar to the one in categorial grammars.
Propositions in the context-free framework
are of the form $\beta(u)$,
where $\beta \in (\Sigma \cup N)^*$ and $u \in \Sigma^*$.
Intuitively, such a proposition means
that $u$ can be derived from $\beta$ using the rules of the CFG.
Axioms of the calculus of propositions are of the form $a(a)$, $a \in \Sigma$, and $\varepsilon(\varepsilon)$,
and the rules of inference are as follows.
\begin{align*}
&\infer{(\beta_1 \beta_2)(u_1 u_2)}{\beta_1(u_1) & \beta_2(u_2)}
	\\
&\infer{A(v)}{\beta(v)}
	&& \text{for each rule $A \to \beta$}
\end{align*}

Again, the string $w$ belongs to the language generated by this grammar if and only if
the proposition $S(w)$ is derivable.

\begin{example}
The language from Example~\ref{BCG_example} is described by the following CFG.
\begin{align*}
S &\to b A \\
A &\to aAa \ | \ c
\end{align*}
As usual, ``$A \to aAa \ |\ c$'' is a short-hand notation for two rules,
$A \to aAa$ and $A \to c$.
\end{example}

There is an important difference between BCGs and CFGs: in BCGs, the linguistic information
is stored in the axioms, each associated with an alphabet symbol
(in other words, it is {\em lexicalized}), while the inference rules are
the same for all BCGs.
For CFGs, the situation is opposite:
axioms are trivial,
and all information is kept in the rules.
However, these two formalisms are equivalent in power, provided the languages in question do not include $\varepsilon$.

\begin{oldtheorem}[Bar-Hillel et al.~\cite{BGS1960}]\label{Th:BarHillel}
A language is generated by a BCG if and only if it is generated by a CFG and does not include $\varepsilon$.
\end{oldtheorem}



\section{Conjunction in Grammars}\label{section_definitions_conjunction}

In this section, both grammar formalisms
are enriched with a \emph{conjunction} operation.
Using conjunction, one can impose
multiple syntactic constraints
on the same phrase at the same time.
The extension of context-free grammars with conjunction,
called {\em conjunctive grammars,}
was introduced by Okhotin~\cite{Conjunctive}.


Let $\Sigma$ be the alphabet,
and let $N$ be the set of categories (``non-terminal symbols''),
with $S \in N$ representing all well-formed sentences.
A conjunctive grammar is defined by a finite set of rules
of the form $A \to \beta_1 \C \dots \C \beta_k$,
with $\beta_i \in (\Sigma \cup N)^*$.
If $k$ is $1$, then this is an ordinary rule $A \to \beta$,
as in an ordinary context-free grammar.

Propositions in a conjunctive grammar are of the form
$\beta(u)$,
with $u \in \Sigma^*$
and $\beta \in (\Sigma \cup N)^*$,
the same as for ordinary context-free grammars.

Axioms of the calculus of propositions are of the form
$a(a)$, where $a \in \Sigma$, and $\varepsilon(\varepsilon)$,
and the inference rules
are as follows.
\begin{align*}
	&\infer{(\beta_1 \beta_2)(u_1 u_2)}{\beta_1(u_1) & \beta_2(u_2)}
\intertext{%
The other inference rule is valid
for each grammar rule $A \to \beta_1 \C \ldots \C \beta_k$
and for each string $v$.
}
	&\infer{A(v)}{\beta_1(v) & \dots & \beta_k(v)}
\end{align*}

The string $w$ belongs to the language generated by the grammar
if and only if
the proposition $S(w)$ is derivable from the axioms.

As for context-free grammars, for languages without $\varepsilon$ one may, w.l.o.g., suppose that the conjunctive
grammar does not have $\varepsilon$-rules, i.e., for each rule $A \to \beta_1 \C \dots \C \beta_k$ we have $\beta_1, \ldots, \beta_k \in (\Sigma \cup N)^+$.

\begin{example}\label{conjunctive_example}
The following conjunctive grammar
describes the language $\{ b a^n c a^n c a^n \mid n \geqslant 1 \}$.
\begin{align*}
	S &\to b B c A \mathop{\&} b A c B \\
	A &\to a A \ | \ a \\
	B &\to a B a \ | \ c
\end{align*}
The rules for $A$ and $B$ use no conjunction,
and have the same effect
as in ordinary context-free grammars.
Thus, $bBcA(w)$ is true
for all strings of the form $w=b a^n c a^n c a^i$,
with $n \geqslant 0$, $i \geqslant 1$,
whereas $bAcB(w)$ holds true
for strings of the form $w=b a^i c a^n c a^n$.
The conjunction of these two conditions
is exactly the condition of membership in the desired language,
and the rule for $S$ ensures it
by derivations of the following form.
\begin{equation*}
	\infer{S(b a^n c a^n c a^n)}{bBcA(b a^n c a^n c a^n) & bAcB(b a^n c a^n c a^n)}
\end{equation*}
For the string $w=bacaca$,
the full derivation of its syntactical correctness is given below.
\begin{align*}
	\infer{S(bacaca)}{
		\infer{bBcA(bacaca)}{
			b(b)
			&
			\infer{B(aca)}{
				a(a) & B(c) & a(a)
			}
			&
			c(c)
			&
			A(a)
		}
		&
		\infer{bAcB(bacaca)}{
			b(b)
			&
			A(a)
			&
			c(c)
			&
			\infer{B(aca)}{
				a(a) & B(c) & a(a)
			}			
		}
	}
\end{align*}
\end{example}


The notion of a \emph{conjunctive categorial grammar}
is defined by extending basic categorial grammars
with the conjunction operation.
Let $\PC = \{ p, q, r, \dots \}$ be the set of primitive categories,
$s \in \PC$ is the target category.
\\
The set of {\em conjuncts,} $\Conj$, is defined as the smallest set satisfying the following:
\begin{enumerate}
\item
	every primitive category is a conjunct;
\item
	if $p_1, \dots, p_k \in \PC$, then $(p_1 \wedge \dots \wedge p_k) \in \Conj$.
\end{enumerate}
\noindent
The set of {\em basic categories with conjunction,} $\Cat$, is defined as the smallest set satisfying the following:
\begin{enumerate}
\item
	Every primitive category belongs to $\Cat$.
\item
	If $\conjcat \in \Conj$ and $A \in \Cat$,
	then $(\conjcat \BS A) \in \Cat$
	and $(A \SL \conjcat) \in \Cat$.
\end{enumerate}
\noindent
{\em Categorial propositions} are expressions of the form
$B(v)$, where $v \in \Sigma^+$ and $B \in \Cat \cup \Conj$.
A conjunctive categorial grammar is a logical theory deriving categorial
propositions.
It includes an arbitrary finite set of axioms
of the form $A(a)$, with $A \in \Cat$ and $a \in \Sigma$,
and the following inference rules.
$$
\infer{(p_1 \wedge \dots \wedge p_k)(v)}{p_1(v) & \dots & p_k(v)}
$$ $$
\infer{A(uv)}{\conjcat(u) & (\conjcat \BS A)(v)}
\qquad
\infer{A(vu)}{(A \SL \conjcat)(v) & \conjcat(u)}
$$

The string $w$ belongs to the language generated by this grammar if and only if
the proposition $s(w)$ is derivable.

\begin{example}\label{conjCG_example}
The conjunctive categorial grammar
with the set of primitive categories $\PC=\{s, x, y, p, q, r \}$,
$s$ as the target category, 
and with the following set of axioms,
describes the language $\{ b a^n c a^n c a^n \mid n \geqslant 1 \}$,
the same as in Example~\ref{conjunctive_example}.
\begin{align*}
& r(a), \qquad (r \SL r)(a), \\
& p(c), \qquad (p \SL q)(a), \qquad (p \BS q)(a),\\
& (p \BS (x \SL r))(c), \qquad ((r \BS y) \SL p)(c),\\
& (s \SL (x \wedge y))(b).
\end{align*}
\begin{figure*}
\small
\begin{equation*}
\infer{s(bacaca)}{
	(s \SL (x \wedge y))(b)\hspace*{-2.5cm}
	&
	\infer{(x \wedge y)(acaca)}{
		\infer{x(acaca)}{
			\infer{(x \SL r)(acac)}{
				\infer{p(aca)}{
					(p \SL q)(a)
					&
					\infer{q(ca)}{
						p(c)
						&
						(p \BS q)(a)
					}
				}
				&
				\hspace*{-.8cm}
				(p \BS (x \SL r))(c)
			}
			&
			r(a)
		}
		&
		\infer{y(acaca)}{
			r(a)
			&
			\infer{(r \BS y)(caca)}{
				((r \BS y) \SL p)(c)
				&
				\infer{p(aca)}{
					(p \SL q)(a)
					&
					\infer{q(ca)}{
						p(c)
						&
						(p \BS q)(a)
					}
				}
			}
		}
	}
}
\end{equation*}
\caption{}\label{Fig:conjunctive_categorial_derivation_example}
\end{figure*}%
The category $p$ is defined in the same way as in Example~\ref{BCG_example}.
Then, using further categories without conjunction,
the propositions $x(a^n c a^n c a^i)$
and $y(a^i c a^n c a^n)$,
for all $n \geqslant 0$, $i \geqslant 1$,
are derived as in an ordinary categorial grammar.
The only strings that satisfy both conditions, $x$ and $y$,
are those of the form $a^n c a^n c a^n$,
and these are therefore all strings in the category $s$,
derived as follows.
\begin{equation*}
	\infer{s(b a^n c a^n c a^n)}{(s \SL (x \land y))(b) &
		\infer{(x \land y)(a^n c a^n c a^n)}{x(a^n c a^n c a^n) & y(a^n c a^n c a^n)}
	}
\end{equation*}
A complete derivation of the proposition $s(bacaca)$
is presented in Figure~\ref{Fig:conjunctive_categorial_derivation_example}.
\end{example}

The calculus used in the conjunctive categorial grammar formalism
enjoys the following \emph{inverted subformula property} (ISF):
if a category of the form $(\conjcat \BS A)$ or $(A \SL \conjcat)$
appears somewhere in the derivation,
then it is a subexpression of some category used in an axiom.
(The notion of subexpression on categories is defined in a standard way:
each conjunct (in particular, primitive category) is a subexpression of itself, and subexpressions
of $(\conjcat \BS A)$ include $\conjcat \BS A$, $\conjcat$, and all subexpressions of $A$;
symmetrically for $(A \SL \conjcat)$. To prove the ISF, we trace the rightmost branch of the derivation
upwards; finally we reach an axiom that includes the goal category as a subexpression.)

Another useful property is the fact that the rule for $\wedge$ is invertible:
if $(p_1 \wedge \ldots \wedge p_k)(v)$ is derivable,
then so are $p_1(v)$, \dots, $p_k(v)$.
Indeed, the only way to derive
$(p_1 \wedge \ldots \wedge p_k)(v)$
is by applying this rule.

The calculus used in conjunctive categorial grammars
also enjoys the following \emph{cut elimination} property.

\begin{lemma}\label{Lm:conjCG_cutelim}
Let $A(u)$ (for some $A \in \Cat$ and $u \in \Sigma^+$) be derivable in the given conjunctive
categorial grammar. Consider a new conjunctive categorial grammar over an extended alphabet
$\Sigma \cup \{ b \}$, where $b \notin \Sigma$.
The new grammar has all the same axioms as the original grammar,
and an additional axiom $A(b)$.
Then, if the new grammar derives $B(v_1 b v_2)$,
for some $B\in \Cat$
and arbitrary, possibly empty, strings $v_1$, $v_2$ over $\Sigma$,
then $B(v_1 u v_2)$ is derivable in the original grammar.
\end{lemma}
\begin{proof}
Consider the derivation of $B(v_1 b v_2)$ in the extended grammar and substitute $u$ for
all occurrences of $b$. Applications of inference rules remain valid; the same for axioms
of the old grammar (they don't include $b$). The new axiom $A(b)$ becomes $A(u)$, which
is derivable in the old grammar by assumption.
\end{proof}

\section{Equivalence of Conjunctive Grammars and Conjunctive Categorial Grammars}\label{section_equivalence}

The main result of this paper
is an extension of Theorem~\ref{Th:BarHillel}
for grammars with conjunction,
stated as follows.

\begin{theorem}\label{conjunctive_iff_conjunctive_categorial_theorem}
A language without the empty string is generated by a conjunctive grammar if and only if it is generated by
a conjunctive categorial grammar.
\end{theorem}

The proof uses the following two known properties of conjunctive grammars.
The first result is their closure under quotient with a single symbol.

\begin{oldlemma}[{Okhotin and Reitwie{\ss}ner~\cite[Thm.~2]{OkhotinReitwiessner}}]\label{Lm:leftfactor}
If $L$ is a language over $\Sigma$ described by a conjunctive grammar,
and $a \in \Sigma$ is any symbol,
then there exists a conjunctive grammar
that describes the language
$a^{-1} L = \{w \mid aw \in L\}$. Moreover, if $\varepsilon \in a^{-1} L$, then there also exists a conjunctive
grammar for the language $a^{-1} L - \{ \varepsilon \}$.
\end{oldlemma}

The other result is a normal form theorem.
A conjunctive grammar $G$ with the initial symbol $S$
is in the \emph{odd normal form},
if all its rules are of the following form.
\begin{align*}
	A &\to a
		&& (a \in \Sigma)
		\\
	A &\to B_1 a_1 C_1 \C \ldots \C B_k a_k C_k
		&& (B_i,C_i \in N, \: a_i \in \Sigma)
		\\
	S &\to a A
		&& (a \in \Sigma, \: A \in N)
\end{align*}
Rules of the latter kind
are allowed only if $S$ is never referenced in any rules.

\begin{oldtheorem}[Okhotin and Reitwie{\ss}ner~\cite{OkhotinReitwiessner}]\label{odd_normal_form_theorem}
Every language without the empty string described by a conjunctive grammar
can be described
by a conjunctive grammar in the odd normal form.
\end{oldtheorem}

\begin{proof}[Proof of Theorem~\ref{conjunctive_iff_conjunctive_categorial_theorem}]
{\bf The ``if'' part} is easier.
For a given conjunctive categorial grammar $G$,
the equivalent conjunctive grammar $G'$
has the set $N$ comprised of all categories
used in the axioms of $G$,
and of all their subexpressions (categories and conjuncts).
The conjunctive rules are now as follows.
\begin{align*}
	(p_1 \wedge \ldots \wedge p_k) &\to p_1 \C \ldots \C p_k
		&& ((p_1 \wedge \ldots \wedge p_k) \in N)
	\\
	A &\to \conjcat \, (\conjcat \BS A)
		&& ((\conjcat \BS A) \in N)
		\\
	A &\to (A \SL \conjcat)\, \conjcat
		&& ((A \SL \conjcat) \in N)
			\\
	A &\to a
		&& (A(a) \text{ is an axiom in } G)
\end{align*}
Then, $G$ and $G'$ have the same set of derivations,
and accordingly, a proposition $A(w)$ is derivable in $G'$
if and only if $A(w)$ is derivable in $G$.

For {\bf the ``only if'' part} of the proof,
the first step is to transform a given conjunctive grammar.
Let $\Sigma = \{ a_1, \dots, a_n \}$.
For each symbol $a_i$, consider the quotient $a_i^{-1} L$. Unlike $L$ itself, $a_i^{-1} L$ could
potentially include the empty string. This happens if the one-letter string $a_i$ belongs to $L$.
Let us temporarily remove $\varepsilon$ from $a_i^{-1} L$, and construct a conjunctive grammar
$G_i$ for the language $a_i^{-1} L - \{ \varepsilon \}$. Such a grammar exists by Lemma~\ref{Lm:leftfactor}.

By Theorem~\ref{odd_normal_form_theorem},
the grammar $G_i$ can be assumed to be in the odd normal form.
(Here we use the fact that we have removed $\varepsilon$.)
It can also be assumed that,
for $i\ne j$, the grammars $G_i$ and $G_j$
have disjoint sets of non-terminal symbols.
Let $S_i$ be the initial symbol of $G_i$.
Then this grammar is further modified as follows.
Every rule $$A \to B_1 a_1 C_1 \C
\ldots \C B_k a_k C_k$$
is replaced with $k+1$ new rules:
\begin{equation*}
	A \to \widetilde{X}_1 \C \ldots \C \widetilde{X}_k
		\qquad\mbox{ and }\qquad
	\widetilde{X}_j \to B_j a_j C_j,
\end{equation*}
where $\widetilde{X}_j$ are fresh non-terminals. (For each rule, the $\widetilde{X}_j$'s are different.)
For the sake of uniformity, rules of the form
\begin{equation*}
	S_i \to a A
\end{equation*}
are replaced with
\begin{equation*}
	S_i \to \widetilde{Y}
		\qquad\mbox{ and }\qquad
	\widetilde{Y} \to aA,
\end{equation*}
and rules of the form
\begin{equation*}
	A \to a
\end{equation*}
are replaced with
\begin{equation*}
	A \to \widetilde{Z}
		\mbox{ and }
	\widetilde{Z} \to a.
\end{equation*}

Finally, a new conjunctive grammar for $L$
is obtained by joining these grammars together, for all $i$,
adding the following extra rules
for the new initial symbol $\widetilde{S}$.
\begin{equation*}
	\widetilde{S} \to a_1 S_1,
		\qquad\ldots,\qquad
	\widetilde{S} \to a_n S_n
\end{equation*}
and
\begin{equation*}
    \widetilde{S} \to a_i,\text{ for each $i$ such that $\varepsilon \in a_i^{-1} L$}
\end{equation*}
(The latter set of rules ``restores'' the empty strings removed from the quotients.)

In the resulting grammar, which describes the same language $L$,
all non-terminals are of two types
(with and without a tilde),
and the rules have the following form.
\begin{align*}
	A \to \widetilde{X}_1 \C \ldots \C \widetilde{X}_k
		\qquad
		\mbox{(here $k$ may be 1)}
	\\
	\widetilde{X} \to B a C,
		\qquad
	\widetilde{Y} \to a A,
		\qquad\mbox{ and }\qquad
	\widetilde{Z} \to a
\end{align*}

It is then transformed
to a conjunctive categorial grammar,
with the set of primitive categories
$\PC = \{ p_{\widetilde{X}} \mid \mbox{$\widetilde{X}$ is a non-terminal
decorated with a tilde} \}$,
and with the following axioms.
\begin{enumerate}
\item
	For each rule $\widetilde{Z} \to a$,
	there is an axiom $p_{\widetilde{Z}}(a)$.
\item
	For each pair of rules $\widetilde{Y} \to a A$
	and $A \to \widetilde{X}_1 \C \ldots \C \widetilde{X}_k$,
	the axiom is
	$\Big(p_{\widetilde{Y}} \SL (p_{\widetilde{X}_1} \wedge \ldots \wedge
	p_{\widetilde{X}_k})\Big) (a)$.
\item
	For each triple of rules
	$\widetilde{X} \to B a C$,
	$B \to \widetilde{Y}_1 \C \ldots \C \widetilde{Y}_k$, and
	$C \to \widetilde{Z}_1 \C \ldots \C \widetilde{Z}_m$,
	the axiom is
	\begin{equation*}
		\hspace*{-2em}
		\Big( \big( (p_{\widetilde{Y}_1} \wedge \ldots \wedge
		p_{\widetilde{Y}_k}) \BS p_{\widetilde{X}} \big)  \SL
		(p_{\widetilde{Z}_1} \wedge \ldots \wedge
		p_{\widetilde{Z}_m}) \Big) (a).
	\end{equation*}
\end{enumerate}
The target category is $p_{\widetilde{S}}$.

\begin{claim}
For every non-terminal $\widetilde{X}$ decorated with a tilde,
the proposition $p_{\widetilde{X}}(v)$
is derivable in the newly constructed conjunctive categorial grammar
if and only if
$\widetilde{X}(v)$ is derivable in the original conjunctive grammar.
\end{claim}

\begin{proof}
{\bf The ``if'' part.}
The proof proceeds by induction on the derivation size.
There are three possible cases.

{\it Case 1.} The proposition $\widetilde{X}(v)$ is actually of the
form $\widetilde{Z}(a)$ and is derived from $a(a)$
using the rule $\widetilde{Z} \to a$.
Then $p_{\widetilde{Z}}(a)$ is an axiom in the conjunctive categorial grammar.

{\it Case 2.} The proposition $\widetilde{X}(v)$ is of the form
$\widetilde{Y}(av_1)$ and is derived from $a(a)$ and $A(v_1)$ using a rule
of the form $\widetilde{Y} \to aA$. Next, a rule of
the form $A \to \widetilde{X}_1 \C \ldots \C \widetilde{X}_k$ should be
applied for $A$. Therefore, the propositions $\widetilde{X}_i(v_1)$ are derivable (for
all $i$) in the original conjunctive grammar. Then, by induction hypothesis,
$p_{\widetilde{X}_i}(v_1)$ are derivable in the
conjunctive categorial grammar,
and there is a derivation for $p_{\widetilde{Y}}(av_1)$
shown in Figure~\ref{Fig:first}.

\begin{figure*}
$$
\infer{p_{\widetilde{Y}}(av_1)}
{\bigl(p_{\widetilde{Y}} \SL (p_{\widetilde{X}_1} \wedge \ldots \wedge
p_{\widetilde{X}_k})\bigr) (a) &
\infer{(p_{\widetilde{X}_1} \wedge \ldots \wedge
p_{\widetilde{X}_k})(v_1)}{p_{\widetilde{X}_1}(v_1) & \ldots & p_{\widetilde{X}_k}(v_1)}
}
$$
\caption{}\label{Fig:first}
\end{figure*}

{\it Case 3.}
The proposition $\widetilde{X}(v)$ is of the form $\widetilde{X}(v_1av_2)$ and is derived
from some propositions of the form $B(v_1)$, $a(a)$, and $C(v_2)$,
following the rule $\widetilde{X} \to BaC$.
Next, for $B$ and $C$,
some rules for the form $B \to \widetilde{Y}_1 \C \ldots \C \widetilde{Y}_k$
and $C \to \widetilde{Z}_1 \C \ldots \C \widetilde{Z}_m$
should be applied. Therefore,
the propositions $\widetilde{Y}_i(v_1)$  and $\widetilde{Z}_j(v_2)$ are derivable
(for all $i$, $j$) in the original conjunctive grammar. Then, by induction hypothesis,
$p_{\widetilde{Y}_i}(v_1)$ and $p_{\widetilde{Z}_j}(v_2)$ are derivable in the
conjunctive categorial grammar,
and there is a derivation for $p_{\widetilde{X}}(v_1 a v_2)$,
as shown in Figure~\ref{Fig:second}.

\begin{figure*}
$$
\infer{p_{\widetilde{X}}(v_1 a v_2)}
{(p_{\widetilde{Y}_1} \wedge \ldots \wedge
p_{\widetilde{Y}_k})(v_1)  \hspace*{-.8em} &
\infer{((p_{\widetilde{Y}_1} \wedge \ldots \wedge
p_{\widetilde{Y}_k}) \BS p_{\widetilde{X}})(av_2)}
{\bigl(((p_{\widetilde{Y}_1} \wedge \ldots \wedge
p_{\widetilde{Y}_k}) \BS p_{\widetilde{X}}) \SL (p_{\widetilde{Z}_1} \wedge
\ldots \wedge p_{\widetilde{Z}_m})\bigr)(a) &
(p_{\widetilde{Z}_1} \wedge
\ldots \wedge p_{\widetilde{Z}_m})(v_2)}}
$$
\caption{}\label{Fig:second}
\end{figure*}

{\bf The ``only if'' part.}
This time, it is assumed that $p_{\widetilde{X}}(v)$
is derivable in the newly constructed conjunctive categorial grammar.
The proof is by induction on its derivation.

The {\em axiom case} is trivial: any axiom of the form $p_{\widetilde{Z}}(a)$
is associated with a rule $\widetilde{Z} \to a$ in the original
conjunctive grammar, and then $\widetilde{Z}(a)$ is derivable from $a(a)$.

In the {\em left division} case, $v = v_1 w$,
and the last step of the derivation is as follows.
\begin{equation*}
\infer{p_{\widetilde{X}}(v_1 w)}
{\conjcat(v_1) & (\conjcat \BS p_{\widetilde{X}})(w)}
\end{equation*}

By the ISF (see above), $(\conjcat \BS p_{\widetilde{X}})$
is a subexpression of the category in one of the axioms.
The only possibility is that $(\conjcat \BS p_{\widetilde{X}})$
is a subexpression
$((p_{\widetilde{Y}_1} \wedge \ldots \wedge p_{\widetilde{Y}_k}) \BS p_{\widetilde{X}})$
of an axiom
$((p_{\widetilde{Y}_1} \wedge \ldots \wedge
p_{\widetilde{Y}_k}) \BS p_{\widetilde{X}}) \SL (p_{\widetilde{Z}_1} \wedge
\ldots \wedge p_{\widetilde{Z}_m})$.
Moreover, again by the ISF,
the only way to derive $(\conjcat \BS p_{\widetilde{X}})(w)$
is to apply the right division rule
to the category used in the axiom.
This analysis shows that the derivation must end
in the way depicted in the earlier Figure~\ref{Fig:second},
where $w=a v_2$.

Since the rules for $\wedge$ are invertible (see above), the propositions
$p_{\widetilde{Y}_1}(v_1)$, \dots, $p_{\widetilde{Y}_k}(v_1)$,
$p_{\widetilde{Z}_1}(v_2)$, \dots, $p_{\widetilde{Z}_m}(v_2)$ are derivable.
By induction hypothesis, $\widetilde{Y}_i(v_1)$ and $\widetilde{Z}_j(v_2)$,
for all $i$, $j$,
are derivable in the original conjunctive grammar.
Then,
the derivation uses the rules
$\widetilde{X} \to BaC$,
$B \to \widetilde{Y}_1 \C \ldots \C \widetilde{Y}_k$,
and $C \to \widetilde{Z}_1 \C \ldots \C \widetilde{Z}_m$,
and is of the following form.
\begin{equation*}
	\infer{\widetilde{X}(v_1 a v_2)}
	{\infer{B(v_1)}{\widetilde{Y}_1(v_1) \; \ldots \; \widetilde{Y}_k(v_1)} & a(a) &
	\infer{C(v_2)}{\widetilde{Z}_1(v_2) \; \ldots \; \widetilde{Z}_m(v_2)}}
\end{equation*}

The {\em right division} case is even easier.
Here a derivation ends as follows.
\begin{equation*}
	\infer{p_{\widetilde{Y}}(wv_1)}{(p_{\widetilde{Y}} \SL \conjcat)(w) & \conjcat(v_1)}
\end{equation*}
By the ISF, the left premise could be nothing but an axiom of the form
$\bigr(p_{\widetilde{Y}} \SL (p_{\widetilde{X}_1} \wedge \ldots \wedge
p_{\widetilde{X}_k})\bigr)(a)$ (and $w = a$).
Then, $\conjcat(v_1)$ is
$(p_{\widetilde{X}_1} \wedge \ldots \wedge p_{\widetilde{X}_k})(v_1)$,
and by the invertibility of the $\wedge$ rule,
all $p_{\widetilde{X}_i}(v_1)$ are derivable.
By the induction hypothesis,
$\widetilde{X}_i(v_1)$, for all $i$,
are derivable in the original conjunctive grammar,
and there is the following derivation for $\widetilde{Y}(av_1)$,
using the rules $\widetilde{Y} \to aA$
and $A \to \widetilde{X}_1 \C \ldots \C \widetilde{X}_k$.
\begin{equation*}
	\infer{\widetilde{Y}(av_1)}{
		a(a)
		&
		\infer{A(v_1)}{
			\widetilde{X}_1(v_1)
			& \ldots &
			\widetilde{X}_k(v_k)
		}
	}
\end{equation*}
\end{proof}

This claim immediately yields the main result, since
$\widetilde{S}(w)$ is derivable in the original conjunctive grammar
if and only if
$p_{\widetilde{S}}(w)$ is derivable in the constructed conjunctive categorial grammar.
\end{proof}

\section{Conjunctive Categorial Grammars and Lambek Grammars with Additives}\label{section_lambek}

Lambek~\cite{Lambek1958} suggested a richer logic as a background for categorial grammars,
called {\em the Lambek calculus}. In the Lambek calculus, or $\LL$ for short, syntactic
categories are built from a set of $\PC = \{ p_1, p_2, p_3, \ldots \}$
of primitive categories using three binary operations:
product ($\cdot$), which means concatenation,
left division ($\BS$),
and right division ($\SL$).
The formal recursive definition is as follows.
\begin{enumerate}
\item
	Every primitive category is a category.
\item
	If $A$ and $B$ are categories,
	then $(A \cdot B)$, $(A \BS B)$, and $(B \SL A)$ are also categories.
\end{enumerate}
The set of all Lambek categories is denoted by $\Tp$.
As opposed to basic categories,
deep nesting of division operations is allowed here,
that is, denominators are allowed to be non-primitive.

A Lambek categorial grammar consists of a target category $s \in \Tp$ (usually $s$ is required
to be a primitive category) and a finite number of axiomatic propositions of the form $A(a)$, where
$A$ is a category and $a$ is a letter of the alphabet.

A string $w = a_1 \ldots a_n$ is considered accepted by the grammar,
if, for some categories $A_1$, \dots, $A_n$,
the propositions $A_i(a_i)$ are included in the grammar as axiomatic ones,
and the {\em sequent} $A_1, \dots, A_n \to s$ is derivable in the Lambek calculus,
which consists of the axioms and inference rules listed in Figure~\ref{Fig:Lambek}.

In Lambek's original formulation, left-hand sides of the sequents in all rules are required to be non-empty. This constraint is called
\emph{Lambek's non-emptiness restriction} and is motivated by linguistic applications of Lambek grammars. In fact, the restriction
really needs to be imposed only on $(\to{\BS})$ and $(\to{\SL})$; all other rules keep left-hand sides non-empty. The version of $\LL$ without
Lambek's non-emptiness restriction is denoted by $\Ls$.

\begin{figure*}
$$
A \to A
$$

$$
\infer[(\BS\to)]
{\Gamma, \Pi, A \BS B, \Delta \to D}
{\Pi \to A & \Gamma, B, \Delta \to D}
\qquad
\infer[(\to\BS)]
{\Pi \to A \BS B}
{A, \Pi \to B}
\qquad
\infer[(\cdot\to)]
{\Gamma, A \cdot B, \Delta \to D}
{\Gamma, A, B, \Delta \to D}
$$

$$
\infer[(\SL\to)]
{\Gamma, B \SL A, \Pi, \Delta \to D}
{\Pi \to A & \Gamma, B, \Delta \to D}
\qquad
\infer[(\to\SL)]
{\Pi \to B \SL A}
{\Pi, A \to B}
\qquad
\infer[(\to\cdot)]
{\Gamma, \Delta \to A \cdot B}
{\Gamma \to A & \Delta \to B}
$$
\caption{Axioms and inference rules in the Lambek Calculus
	(for all $A, B \in \Tp$ and $\Gamma, \Pi, \Delta \in \Tp^*$). For $(\to\SL)$ and $(\to\BS)$, in the case of $\LL$, $\Pi$ is required to be non-empty}\label{Fig:Lambek}
\end{figure*}

Note that in Lambek grammars,
arrows tranditionally point in an opposite direction
than in context-free grammars
($\ldots \to s$ vs. $S \to \ldots$).

The following \emph{cut rule}
is not officially included in the system, but is admissible (Lambek~\cite{Lambek1958}).
\begin{equation*}
	\infer[(\mathrm{cut})]
	{\Gamma, \Pi, \Delta \to D}
	{\Pi \to A & \Gamma, A, \Delta \to D}
\end{equation*}

As one can easily see,
all basic categories, as defined in Section~\ref{S:intro},
are also Lambek categories: $\BCat \subset \Tp$.
Moreover,  as noticed by Buszkowski~\cite{Buszkowski1985},
if a basic categorial grammar is regarded as a Lambek categorial grammar
with the same set of axiomatic propositions,
then it describes the same language.

Next, the Lambek calculus is extended
with the so-called ``additive'' conjunction and disjunction,
as defined by Kanazawa~\cite{Kanazawa1992}.
These new operations correspond to the additive operations
in linear logic by Girard~\cite{Girard1987}.

The extended set of categories, denoted by $\Tp_{\wedge,\vee}$, is defined as the smallest set obeying the following:
\begin{enumerate}
\item $\Pr \subset \Tp_{\wedge,\vee}$;
\item if $A,B \in \Tp_{\wedge,\vee}$, then $(A \cdot B), (A \BS B), (B \SL A), (A \wedge B), (A \vee B) \in \Tp_{\wedge,\vee}$.
\end{enumerate}

Inference rules for the new operations
are depicted in Figure~\ref{Fig:additives}.

\begin{figure*}
$$
\infer[(\wedge\to)_1]
{\Gamma, A_1 \wedge A_2, \Delta \to D}
{\Gamma, A_1, \Delta \to D}
\qquad
\infer[(\wedge\to)_2]
{\Gamma, A_1 \wedge A_2, \Delta \to D}
{\Gamma, A_2, \Delta \to D}
\qquad
\infer[(\to\wedge)]
{\Pi \to A_1 \wedge A_2}
{\Pi \to A_1 & \Pi \to A_2}
$$

$$
\infer[(\vee\to)]
{\Gamma, A_1 \vee A_2, \Delta \to D}
{\Gamma, A_1, \Delta \to D & \Gamma, A_2, \Delta \to D}
\qquad
\infer[(\to\vee)_1]
{\Pi \to A_1 \vee A_2}
{\Pi \to A_1}
\qquad
\infer[(\to\vee)_2]
{\Pi \to A_1 \vee A_2}
{\Pi \to A_2}
$$
\caption{Rules for conjunction and disjunction}\label{Fig:additives}
\end{figure*}


This calculus, denoted by $\MALC$ (``multiplicative-additive Lambek calculus''),
 also enjoys cut elimination and the subformula property.

 The version of $\MALC$ without
 Lambek's non-emptiness restriction is denoted by $\MALCs$. In this section, we focus on
 grammars based on $\MALC$; the case of $\MALCs$ is considered below in Section~\ref{S:empty}.

Lambek categories with $\wedge$ and $\vee$ generalize conjunctive categories (and
conjuncts):
\begin{equation*}
	\Cat \cup \Conj \subset \Tp_{\wedge,\vee},
\end{equation*}
and every conjunctive categorial grammar can be translated into
a Lambek grammar with $\wedge$ and $\vee$.
However, one cannot simply take
the axiomatic propositions of a conjunctive categorial grammar
and use them as axiomatic propositions in the sense of Lambek grammars:
this would yield a grammar that is not equivalent to the original one
(for instance, the Lambek grammar
with the axiomatic propositions from Example~\ref{conjCG_example}
does not accept any strings at all).
The construction has to be more subtle.




\begin{theorem}\label{conjunctive_categorial_to_MALC_transformation_theorem}
Let $\Sigma = \{ a_1, \dots, a_n \}$ and consider a conjunctive categorial grammar
with the following axiomatic propositions:
$$\begin{matrix}
A_{1,1} (a_1), & A_{1,2}(a_1), & \ldots & A_{1,k_1}(a_1),\\
A_{2,1} (a_2), & A_{2,2}(a_2), & \ldots & A_{2,k_2}(a_2),\\
\vdots\\
A_{n,1} (a_n), & A_{n,2}(a_n), & \ldots & A_{n,k_n}(a_n).
\end{matrix}$$
(This notation means that for each letter $a \in \Sigma$ the grammar has $k_i$ axioms and the
categories (belonging to $\Cat$) used in these axioms are denoted by $A_{i,1}, \ldots A_{i,k_i}$.)
Then the Lambek grammar with atomic propositions
$(A_{i,1} \wedge A_{i,2} \wedge \ldots \wedge A_{i,k_i})(a_i)$
(for $i = 1, \ldots, n$) describes the same language as the original
conjunctive categorial grammar. (If $k_i = 1$, we take just $A_{i,1}(a_i)$.)
\end{theorem}
\begin{proof}
Let $B_i = A_{i,1} \wedge A_{i,2} \wedge \ldots \wedge A_{i,k_i}$.
The new Lambek grammar
uses axiomatic propositions of the form $B_i(a_i)$,
one for each symbol in $\Sigma$.
It is sufficient to prove the following:
for the target category $s \in \PC$
and for a string $a_{i_1} \dots a_{i_m}$,
the proposition $s(a_{i_1} \dots a_{i_m})$ is derivable
in the conjunctive categorial grammar
if and only if
the sequent $B_{i_1}, \dots, B_{i_m} \to s$ is derivable in $\MALC$.

{\bf The ``only if'' part.}
In order to use induction
on the length of derivation in the conjunctive categorial grammar,
the statement is proved not only for $s$,
but for an arbitrary category $D \in \Cat \cup\Conj$.

The proof in the base case is immediate:
if $D(a_i)$ is an axiom,
then $D$ is one of the $A_{i,j}$ in the conjunction $B_i$,
and the sequent $B_i \to A_{i,j}$ is derivable
by several applications of the $(\wedge\to)$ rules.

For the induction step, there are three cases.

{\it Case 1:} $D = (p_1 \wedge \ldots \wedge p_k)$.
Then, by the induction hypothesis,
$B_{i_1}, \dots, B_{i_m} \to p_j$ is derivable in $\MALC$ for every $j$,
and $B_{i_1}, \dots, B_{i_m} \to p_1 \wedge \ldots \wedge p_k$
is derived by the $(\to\wedge)$ rule.

{\it Case 2:} $D(a_{i_1}\dots a_{i_m})$ is derived
from $\conjcat(a_{i_1} \dots a_{i_\ell})$
and $(\conjcat \BS D)(a_{i_{\ell+1}} \dots a_{i_m})$ for
some $\conjcat \in \Conj$.
Then, by the induction hypothesis,
the sequents $B_{i_1}, \dots, B_{i_\ell} \to \conjcat$
and $B_{i_{\ell+1}}, \dots, B_{i_m} \to \conjcat \BS D$ are derivable,
and then $B_{i_1}, \dots, B_{i_m} \to D$ can be derived in the following way.
First,
$B_{i_1}, \dots, B_{i_\ell}, \conjcat \BS D \to D$
is derived from $B_{i_1}, \dots, B_{i_\ell} \to \conjcat$ and $D \to D$,
and then it is combined with $B_{i_{\ell+1}}, \dots, B_{i_m} \to \conjcat \BS D$
using the cut rule,
to get $B_{i_1}, \dots, B_{i_\ell}, B_{i_{\ell+1}}, \dots, B_{i_m} \to D$.

{\it Case 3:} $D(a_{i_1} \dots a_{i_m})$
is derived from $(D \SL \conjcat)(a_{i_1} \dots a_{i_\ell})$
and $\conjcat(a_{i_{\ell+1}} \dots a_{i_m})$.
The proof is symmetric.

{\bf The ``if'' part.}
The following more general statement is claimed.
\emph{For every $j = 1, \dots, m$, let $B'_{i_j}$ be a conjunction of an arbitrary subset
of formulae $A_{i_j,k}$ used in the conjunction $B_{i_j}$;
in other words, $B'_{i_j}$ may coincide with $B_{i_j}$ or lack some of the conjuncts.
Then, for any $\conjcat \in \Conj$ (in particular, for $\conjcat = s \in \PC \subset \Conj$),
if  $B'_{i_1}, \dots, B'_{i_m} \to \conjcat$ is derivable in $\MALC$,
then the proposition $\conjcat(a_{i_1} \dots a_{i_m})$
is derivable in the original conjunctive categorial grammar.}

The claim is proved by induction
on the cut-free derivation
of the sequent $B'_{i_1}, \dots, B'_{i_m} \to \conjcat$
in $\MALC$.

{\it Case 1.} $\conjcat = p_1 \wedge \ldots \wedge p_k$, $k \geqslant 2$.
Since the $(\to\wedge)$ rule in $\MALC$ is invertible
(this follows from the cut elimination),
it can be assumed that all $k-1$ applications of this rule
were applied immediately.
\begin{equation*}
	\infer{B'_{i_1}, \dots, B'_{i_m} \to p_1 \wedge \ldots \wedge p_k}{
		B'_{i_1}, \dots, B'_{i_m} \to p_1
		& \ldots &
		B'_{i_1}, \dots, B'_{i_m} \to p_k
	}
\end{equation*}
Then, by the induction hypothesis,
all propositions $p_j(a_1 \ldots a_m)$
are derivable in the conjunctive categorial grammar,
and from them one can derive $(p_1 \wedge \ldots \wedge p_k)(a_1 \dots a_m)$.

In all other cases, $\conjcat \in \PC$.


{\it Case 2:} an axiom. Then, $m = 1$, $B'_{i_1} = \conjcat$,
and, since all elements of $B'_{i_1}$ should be of the form $A_{i_1,k}$, the proposition
$\conjcat(a_1)$ is an axiom of the conjunctive categorial grammar, and is therefore derivable.

{\it Case 3:} the last rule of the derivation is $(\wedge\to)$.
Then, $B'_{i_\ell} = B''_{i_\ell} \wedge A_{i_\ell,k}$:
\begin{align*}
	\infer{B'_{i_1}, \dots, B''_{i_\ell} \wedge A_{i_\ell,k}, \dots, B'_{i_m} \to \conjcat}{
		B'_{i_1}, \dots, B''_{i_\ell}, \dots, B'_{i_m} \to \conjcat
	}
\intertext{%
	or
}
	\infer{B'_{i_1}, \dots, B''_{i_\ell} \wedge A_{i_\ell,k}, \dots, B'_{i_m} \to \conjcat}{
		B'_{i_1}, \dots, A_{i_\ell,k}, \dots, B'_{i_m} \to \conjcat
	}
\end{align*}
In both cases the induction hypothesis is applied:
since $B''_{i_\ell}$ or $A_{i_\ell,k}$ can act as $B'_{i_\ell}$,
the proposition $\conjcat(a_{i_1}\dots a_{i_\ell} \dots a_{i_m})$
is derivable in the conjunctive categorial grammar.





{\it Case 4:} the last rule is $(\BS\to)$.
In this case,
\mbox{$B'_{i_h} = A_{i_h,k} = \conjcat' \BS A'$,}
for some $h$
and for $\conjcat' \in \Conj$ and $A' \in \Cat$,
and the sequent
$B'_{i_1}, \dots, B'_{i_{\ell - 1}}, B'_{i_\ell}, \dots, B'_{i_{h-1}}, \conjcat' \BS A', B'_{i_{h+1}},
\dots, B'_{a_m} \to \conjcat$
is derived from $B'_{i_\ell}, \dots, B'_{i_{h-1}} \to \conjcat'$
and $B'_{i_1}, \dots, B'_{i_{\ell-1}}, A', B'_{i_{h+1}}, \dots, B'_{i_m} \to \conjcat$.
By the induction hypothesis,
the proposition $\conjcat'(a_{i_\ell} \dots a_{i_{h-1}})$
can be derived in the conjunctive categorial grammar,
and, since $(\conjcat' \BS A')(a_{i_h})$ is an axiom,
the proposition $A'(a_{i_\ell} \dots a_{i_{h-1}} a_{i_h})$ is also derivable.

Now, the conjunctive categorial grammar is extended
by adding a new symbol $a_{n+1}$ to the original alphabet $\Sigma = \{ a_1, \dots, a_n \}$,
with a new axiom, $A'(a_{n+1})$.
For the new grammar,
we have the same $B_j$ for $j = 1, \dots, n$, and $B_{n+1} = A'$.
Since $B'_{i_1}, \dots, B'_{i_{\ell-1}}, A', B'_{i_{h+1}}, \dots, B'_{i_m} \to \conjcat$
is derivable in $\MALC$,
by the induction hypothesis,
the proposition $\conjcat(a_{i_1} \dots a_{i_{\ell-1}} a_{n+1} a_{i_{h+1}} \dots a_{i_m})$
is derivable in the extended conjunctive categorial grammar.

By Lemma~\ref{Lm:conjCG_cutelim},
the desired proposition
$\conjcat(a_{i_1} \dots a_{i_{\ell-1}} a_{i_\ell} \dots a_{i_{h-1}} a_{i_h} a_{i_{h+1}} \dots a_{i_m})$,
where the string $u = a_{i_\ell} \dots a_{i_{h-1}} a_{i_h}$
has been substituted for a fresh symbol $b = a_{n+1}$,
can be derived in the original conjunctive categorial grammar.




{\it Case 5:} the last rule is $(\SL\to)$. Symmetric.
\end{proof}

This embedding immediately implies
that every language generated by a conjunctive grammar
can be generated by a $\MALC$-grammar.
This supersedes the result by Kuznetsov~\cite{Kuznetsov2013}, where this statement was established only
for conjunctive grammars in Greibach normal form.

\section{A $\MALC$-grammar Generating an NP-complete Language}

In the classical case without the conjunction,
a converse result was shown by Pentus~\cite{Pentus1993}:
every language generated by a Lambek grammar is context-free.
Whether an analogous property holds for $\MALC$
(that is, whether every $\MALC$-language is generated by a conjunctive grammar)
remains an open problem.
Establishing any such upper bound on the power of the new model
would require proving a non-trivial variant
of the famous theorem by Pentus~\cite{Pentus1993},
which would likely be difficult.

However, there is some evidence that $\MALC$
should be strictly more powerful than conjunctive grammars.
First, there is a result by Okhotin~\cite{BooleanPComplete}
that conjunctive grammars can describe a certain P-complete language
representing the Circuit Value Problem (CVP) under a suitable encoding.
On the other hand,
the class of languages generated by $\MALC$-grammars
is, by definition, closed under symbol-to-symbol homomorphisms.
These two facts are sufficient
to develop a $\MALC$ represenation for an NP-complete language.

\begin{theorem}\label{np_complete_as_h_of_conj_theorem}
There exists a conjunctive grammar $G=(\Sigma, N, R, S)$,
and a length-preserving homomorphism $h \colon \Sigma \to \Omega$,
for which $h(L(G))$ is an NP-complete language.
\end{theorem}

The proof of the theorem
is based upon representing the Circuit Value Problem (CVP)
by a conjunctive grammar.
A circuit is a sequence of gates $C_1, \ldots, C_n$,
with a Boolean value computed at each gate.
The first $m$ gates are \emph{input values}, each set to 0 or to 1.
Each of the following gates $C_i$, with $m < i \leqslant n$,
computes some Boolean function,
with any of the earlier gates $C_1, \ldots, C_{i-1}$ as arguments.
In its classical form, the CVP is stated as follows:
``Given a circuit,
with a conjunction or a disjunction of two arguments or a negation of one argument
computed in each gate,
and given a vector of input values,
determine whether the last, $n$-th gate evaluates to 1''.
This is the fundamental P-complete problem.

Many of its variants are known to be P-complete as well.
Among these, the variant particularly suitable
for representation by conjunctive grammars
is the \emph{Sequential NOR Circuit Value Problem} (Seq-NOR-CVP),
where the function computed in each gate
is $C_i=\lnot(C_{i-1} \lor C_j)$,
with the previous gate $C_{i-1}$ as one of the arguments
and with any earlier gate $C_j$, with $j < i$, as the other argument.
Every standard circuit with conjunction, disjunction and negation
can be translated to this form,
with the set of the inputs preserved,
and with the same value computed in the output gate,
see Okhotin~\cite{BooleanPComplete}.
Accordingly, Seq-NOR-CVP remains P-complete.

\begin{lemma}[{Okhotin~\cite[Thm.~3]{BooleanPComplete}}]\label{sequential_nor_cvp_by_conjunctive_grammar_lemma}
Consider the following encoding of Sequential NOR circuits
as strings over the alphabet $\Sigma=\{0, 1, a, b\}$.
Each gate $C_i$ is represented by a substring of the following form:
$0$ for $C_i=0$,
$1$ for $C_i=1$,
$0$ for $C_i=0$,
$a^{n-j-1} b$ for $C_i=\lnot(C_{i-1} \lor C_j)$.
The whole circuit is then represented
as a concatenation of these substrings,
in the reverse order from $C_n$ to $C_1$.

Then there exists a conjunctive grammar
that describes the encodings of all circuits that evaluate to 1.
\end{lemma}
\begin{proof}
Under this encoding,
the value of a circuit can be determined as follows.
\begin{itemize}
\item
	A circuit $0 w$, where $w$ is a circuit, has value 0.
\item
	A circuit $1 w$, where $w$ is a circuit, has value 1.
\item
	A circuit $a^m b w$, where $w$ is a circuit, has value 1
		if and only if
	$w$ has value 0, \emph{and}
	$w$ is in $(a^* b \cup \{0, 1\})^m u$,
	where $m \geqslant 0$ and $u$ has value 0.
\end{itemize}
This definition is implemented in the following conjunctive grammar,
where the nonterminals $T$ and $F$
define all well-formed encodings of circuits that evaluate to 1 and to 0, respectively.
\begin{align*}
	T &\to A b F \mathop{\&} C F \ | \ 1T \ | \ 1F \ | \ 1 \\
	F &\to A b T \ | \ C T \ | \ 0T \ | \ 0F \ | \ 0 \\
	A &\to a A \ | \ \epsilon \\
	C &\to a C A b \ | \ aC0 \ | \ aC1 \ | \ b
\end{align*}
The rules for $C$ define all prefixes
of the form $a^m x_1 \ldots x_m$,
with each $x_i$ in $a^* b \cup \{0, 1\}$.
\end{proof}

\begin{proof}[Proof of Theorem~\ref{np_complete_as_h_of_conj_theorem}]
The desired grammar $G$
is given by Lemma~\ref{sequential_nor_cvp_by_conjunctive_grammar_lemma},
it is defined over the alphabet $\Sigma=\{a, b, 0, 1\}$.
Let $\Omega=\{a, b, ?\}$
and let $h \colon \Sigma \to \Omega$
be a homomorphism that maps both digits to the question mark symbol,
leaving all other symbols intact:
$h(0)=h(1)={?}$, $h(a)=a$, $h(b)=b$.
This transforms the Circuit Value Problem
to the Circuit Satisfiability Problem,
which is NP-complete.
\end{proof}

\begin{corollary}
The family of languages generated by $\MALC$-grammars contains an $\mathrm{NP}$-complete language.
\end{corollary}
\begin{proof}
Let $G$ be the grammar
and let $h$ be the homomorphism given by Theorem~\ref{np_complete_as_h_of_conj_theorem}.
Since the language $L(G)$
is described by a conjunctive grammar,
by Theorem~\ref{conjunctive_iff_conjunctive_categorial_theorem},
it is also described by a conjunctive categorial grammar,
and then, by Theorem~\ref{conjunctive_categorial_to_MALC_transformation_theorem},
also by a $\MALC$-grammar.
Next, as observed by Kanazawa~\cite{Kanazawa1992},
its symbol-to-symbol homomorphic image $h(L(G))$
must have a $\MALC$-grammar as well.
\end{proof}

It would be interesting to establish
an unconditional separation of these two classes.
Most likely, this is much easier than to separate P from NP,
because conjunctive grammars should have some limitations
that make certain relatively simple languages non-representable.
However, no such limitations are yet known,
as no method for proving non-existence of conjunctive grammars
for particular languages have been discovered,
see a survey by Okhotin~\cite{BooleanSurvey} for a detailed discussion.
This lack of methods is currently the main open question
in the study of conjunctive grammars,
and any advances on this question
will likely help separating $\MALC$-grammars
from conjunctive grammars.

Returning to $\MALC$ itself, the derivability problem in this calculus is known to be
PSPACE-complete. (This result was presented by Kanovich at the CSL 1994 conference; for the proof, see
Kanovich et al.~\cite{Kanovich2019WoLLIC}.) For languages generated by \emph{concrete} $\MALC$-grammars,
however, by now we only know an NP-hard one (see above). It remains an open problem whether there exists
a $\MALC$-grammar generating a harder language (e.g., a PSPACE-complete one). Interestingly enough, an
there is an analogous gap for the original Lambek grammar formalism. On one hand, as shown by Pentus~\cite{Pentus1993},
every language generated by a fixed Lambek grammar is context-free, thus polynomially decidable. On the other hand,
the derivability problem in the Lambek calculus $\mathbf{L}$ is NP-complete (Pentus~\cite{Pentus2006}).


\section{Adding the Empty String}\label{S:empty}

Unlike conjunctive categorial grammars and $\MALC$-grammars, grammars based on $\MALCs$ are capable of describing languages with the empty string~$\varepsilon$. Indeed, in $\MALCs$ one can derive sequents with empty antecedents, e.g., ${} \to p \SL p$.

In this section, we are going to extend Theorem~\ref{conjunctive_categorial_to_MALC_transformation_theorem} and show that even if a language generated by a conjunctive grammar includes $\varepsilon$, it can still be generated by an appropriate $\MALCs$-grammar. First, we notice that the construction in Theorem~\ref{conjunctive_categorial_to_MALC_transformation_theorem} does not depend on Lambek's non-emptiness restriction: if one replaces $\MALC$ with $\MALCs$, the language remains the same. However, for languages including $\varepsilon$ some extra work should be done, and in this section we show how to do it.

The construction is essentially the same as the one used by Kuznetsov~\cite{Kuzn2012IGPL} to build Lambek grammars for context-free languages with the empty string. However, the proof technique is different, based on direct derivation analysis rather than on proof nets (the version of which used in the aforementioned article does not extend to additives).

\begin{theorem}\label{MALCs_empty_string}
 Any language generated by a conjunctive grammar can be generated by a $\MALCs$-grammar.
\end{theorem}

As noticed above, the interesting case here is the empty string case. Let the language in question be $L \cup \{ \varepsilon \}$, where $L$ is a language without the empty string, generated by a conjunctive grammar. By Theorem~\ref{conjunctive_iff_conjunctive_categorial_theorem}, $L$ is generated by a conjunctive categorial grammar, and Theorem~\ref{conjunctive_categorial_to_MALC_transformation_theorem} yields a Lambek grammar $G$, which generates $L$ both over $\MALC$ and over $\MALCs$. The latter is what we shall use. Recall that the goal category of this Lambek grammar is a primitive one, namely, $s = p_{\widetilde{S}}$.

Now it remains to modify the grammar so that is accepts the empty string and exactly those non-empty strings which were accepted by the original grammar. This is performed by replacing each occurrence of the primitive category $p_{\widetilde{S}}$ with the following formula:
\[
 D = \bigl( (r\BS r) \BS ((t \BS t) \BS q) \bigr) \BS q,
\]
where $q,r,t$ are fresh variables (primitive
categories).

We claim that this substitution yields the desired $\MALCs$-grammar which generates the language $L \cup \{ \varepsilon \}$. Our argument will heavily depend on the format of the original grammar $G$, i.e., that it appeared by construction coming from Theorems~\ref{conjunctive_iff_conjunctive_categorial_theorem} and~\ref{conjunctive_categorial_to_MALC_transformation_theorem}.

It is easy to see that the sequent ${}\to D$, with the empty antecedent, is derivable in $\MALCs$ (see Kuznetsov~\cite[Lem. 12(1)]{Kuzn2012IGPL}):
\[
\infer[(\to{\BS})]
{{} \to \bigl( (r\BS r) \BS ((t \BS t) \BS q) \bigr) \BS q}
{\infer[({\BS}\to)]{(r\BS r) \BS ((t \BS t) \BS q) \to q}
{\infer[(\to{\BS})]{{} \to r \BS r}{r \to r} & \infer[({\BS}\to)]{(t \BS t) \BS q \to q}{\infer[(\to{\BS})]{{} \to t \BS t}{t \to t} & q \to q}}}
\]
Therefore, $\varepsilon$ belongs to the language generated by our grammar.

Moreover, substituting a complex category for a primitive one ($D$ for $p_{\widetilde{S}}$; here and further we denote the substitution by $[p_{\widetilde{S}} := D]$) in a derivable sequent keeps it derivable. Hence, any string from $L$, being generated by the original grammar, is also generated by the new one. Thus, it remains to show that the new grammar does not generate any extra non-empty string, which does not belong to $L$.
In other words, for any non-empty $\Pi$, which is a sequence of formulae associated to letters by grammar~$G$, we have to prove that if the sequent $\Pi[p_{\widetilde{S}} := D] \to D$ is derivable in $\MALCs$, then so is $\Pi \to p_{\widetilde{S}}$.

For brevity, let us introduce the following notation. Let $E = (r\BS r) \BS ((t \BS t) \BS q)$; thus, $D = E \BS q$. Throughout the rest of this section, $\Pi$, possibly with subscripts, will denote sequences of formulae of the following forms:
\begin{enumerate}
\item $p_{\widetilde{Z}}$;
\item $p_{\widetilde{Y}} \SL (p_{\widetilde{X}_1} \wedge \ldots \wedge p_{\widetilde{X}_k})$, where $\widetilde{X}_i \ne \widetilde{S}$ for each $i$;
\item $((p_{\widetilde{Y}_1} \wedge \ldots \wedge p_{\widetilde{Y}_k}) \BS p_{\widetilde{X}}) \SL (p_{\widetilde{Z}_1} \wedge \ldots \wedge p_{\widetilde{Z}_m})$, where $\widetilde{Y}_i \ne \widetilde{S}$ and
$\widetilde{Z}_j \ne \widetilde{S}$ for each $i,j$;
\item conjunctions of such formulae.
\end{enumerate}
We shall call such formulae {\em simple} ones.
Notice that all formulae which $G$ associates to letters are simple. (In particular, in our translation of a conjunctive grammar to $G$, the non-terminal symbol $\widetilde{S}$ is never used to the left of $\to$, whence $p_{\widetilde{S}}$ never appears under division.) A formula will be called {\em $D$-simple,} if it is obtained from a simple formula by the $[p_{\widetilde{S}} := D]$ substitution.

Let $\Pi^D$ denote $\Pi[p_{\widetilde{S}} := D]$.

We are going to analyse the derivation of $\Pi^D \to D$, aiming to reconstruct  a derivation of $\Pi \to \widetilde{p}_S$ from it (which will accomplish our goal). We shall use the techniques of derivation analysis in the Lambek calculus developed in the works of Safiullin~\cite{Safiullin2007} and Kuznetsov~\cite{Kuznetsov2021RSL}.

First, we recall that the $(\to{\BS})$ rule is invertible using cut:
\[
\infer[(\mathrm{cut})]
{A, \Pi \to B}
{\Pi \to A \BS B & \infer[({\BS}\to)]{A, A \BS B \to B}{A \to A & B \to B}}
\]
Thus, our sequent $\Pi^D \to D$, where $D = E \BS q$, is equiderivable with $E, \Pi^D \to q$.

Second, let us trace the cut-free derivation upwards from the goal sequent, going to the right (``main'') premise at each application of $({\BS}\to)$ and $({\SL}\to)$. Besides $({\BS}\to)$ and $({\SL}\to)$, the only other rule which may appear on this trace is $({\wedge}\to)$. Applications of $(\to{\wedge})$ may also occur in the derivation, but not on the given trace. The $({\wedge}\to)$ rule does not branch the trace. Therefore, the trace reaches exactly one instance of the $q \to q$ axiom. Now we track the left $q$ downwards, to an occurrence of $q$ in the antecedent $E, \Pi^D$. This occurrence will be called the {\em principal} occurrence of $q$. Notice that the principal occurrence depends on the concrete derivation, not just on the fact of derivability.

The following lemma is proved exactly as Lemma~5.19 of Kuznetsov~\cite{Kuznetsov2021RSL}, by induction on the cut-free derivation. Adding conjunction does not significantly change the argument.

\begin{lemma}\label{Lm:principal}
Let each formula in $\Phi$ and $\Psi$ be $E$, $t\BS t$, $r \BS r$, or a $D$-simple one. Let $C$ be either equal to $A_1 \BS \ldots \BS A_k \BS q \SL B_1 \ldots \SL B_m$ or be a conjunction including $A_1 \BS \ldots \BS A_k \BS q \SL B_1 \ldots \SL B_m$ along with some other $D$-simple formulae.
(Here $\BS$ associates to the right and $\SL$ associates to the left.) Finally, let the sequent
\[ \Phi, C, \Psi \to q \]
be derivable in $\MALCs$, with the central occurrence of $q$ in the formula $A_1 \BS \ldots \BS A_k \BS q \SL B_1 \ldots \SL B_m$ inside $C$ being the principal one. Then
\(\Phi = \Phi_k, \ldots, \Phi_1,\)
\(\Psi = \Psi_m, \ldots, \Psi_1,\)
and the following sequents are derivable:
\begin{align*}
& \Phi_1 \to A_1,
\qquad \ldots,
\qquad \Phi_k \to A_k, \\
& \Psi_1 \to B_1,
\qquad \ldots,
\qquad \Psi_m \to B_m.
\end{align*}
\end{lemma}

In particular, if $m = 0$, then $\Psi$ should be empty, and symmetrically for $k$ and $\Phi$.

Let us prove several auxiliary statements.

\begin{lemma}\label{Lm:GammaN}
Let $\Gamma = t \BS t, r \BS r$ and let $\Gamma^n$ be the sequence of $n$ copies of $\Gamma$. Then the sequent $\Gamma^n, E, \Pi^D \to q$ is not derivable if $n \geqslant 2$.
\end{lemma}

\begin{proof}
We proceed by induction on the number of formulae in $\Pi$.
Let us suppose the contrary and  locate the principal occurrence of $q$ and consider two cases.

 {\em Case 1:} the principal occurrence of $q$ is in $E$. Then by Lemma~\ref{Lm:principal} $\Pi$ is empty, and we have the following derivable sequents: $\Phi_1 \to t \BS t$ and $\Phi_2 \to s \BS s$, where $\Phi_1, \Phi_2 = \Gamma^n$.
 The only two subsequences $\Phi_1$ of $\Gamma^n$, which make $\Phi_1 \to t \BS t$ derivable, are the empty one and $t \BS t$ itself; similarly, the only two options for $\Phi_2$ are empty or $r \BS r$. This yields $n = 0$ or $n = 1$. Contradiction.

 {\em Case 2:} the principal occurrence of $q$ is somewhere in $\Pi^D$. In this case the formula where this principal occurrence is located is of one of the following forms: $E \BS q$, $E \BS q \SL C_1$, or $C_2 \BS E \BS q \SL C_1$, where $C_i$ are conjunctions of variables. By Lemma~\ref{Lm:principal}, we have a derivable sequent of the form $\Phi_1 \to E$, where $\Phi_1$ is a prefix of $\Gamma^n, E, \Pi^D$.

 Moreover, $\Phi_1$ should include $E$, because otherwise the $q$ occurrence in $E$ on the right would not have a match on the left. Thus, we have a derivable sequent of the form $\Gamma^n, E, \Pi_1^D \to E$. Recalling that $E = (r \BS r) \BS ((t \BS t) \BS q)$ and inverting $(\to\BS)$, we get derivability of $\Gamma^{n+1}, E, \Pi_1^D \to q$.

 Since $\Pi_1$ is shorter than $\Pi$, the induction hypothesis is applicable. Having $n+1 > n \geqslant 2$, we conclude that this sequent is not derivable. Contradiction.
\end{proof}

\begin{lemma}\label{Lm:emptyPi}
 If $E, \Pi^D \to E$ is derivable, then $\Pi$ is empty.
\end{lemma}

\begin{proof}
 By inverting $(\to{\BS})$, we get derivability of $t \BS t, r \BS r, E, \Pi \to q$. Let us locate the principal occurrence of $q$.

 {\em Case 1:} the principal occurrence is in $E$. Then $\Pi$ is empty by Lemma~\ref{Lm:principal}.

 {\em Case 2:} the principal occurrence is in $\Pi^D$.  As in Case~2 of the previous lemma, we have a derivable sequent of the form $t \BS t, r \BS r, E, \Pi_1^D \to E$. Inverting $(\to{\BS})$ yields derivability of $t \BS t, r \BS r, t \BS t, r \BS r, E, \Pi_1^D \to q$. However, this sequent is not derivable by Lemma~\ref{Lm:GammaN} (with $n=2$). Contradiction.
\end{proof}

\begin{lemma}\label{Lm:sidesequents}
If $\Pi^D \to p_{\widetilde{Y}_1} \wedge \ldots \wedge p_{\widetilde{Y}_m}$, where $\widetilde{Y}_1 \ne \widetilde{S}$, \ldots, $\widetilde{Y}_m \ne \widetilde{S}$, is derivable, then so is
$\Pi \to p_{\widetilde{Y}_1} \wedge \ldots \wedge p_{\widetilde{Y}_m}$.
\end{lemma}

\begin{proof}
  Though $\Pi$ could include occurrences of $p_{\widetilde{S}}$, we claim that these occurrences are not essential ones. Namely, in the cut-free derivation of $\Pi^D \to p_{\widetilde{Y}_1} \wedge \ldots \wedge p_{\widetilde{Y}_m}$ no occurrence of $D$ (substituted for $p_{\widetilde{S}}$) gets decomposed by $({\BS}\to)$. In other words, each $D$ disappears when $({\wedge}\to)$ is applied, another conjunct being chosen. If we manage to prove this claim, we may replace each $D$ back with $p_{\widetilde{S}}$, keeping the derivation valid. This new derivation proves $\Pi \to p_{\widetilde{Y}_1} \wedge \ldots \wedge p_{\widetilde{Y}_m}$.

  In order to prove the claim, let us suppose the contrary and consider the rightmost $D$ which is decomposed using $({\BS}\to)$ on at least one of the branches. Let us trace the numerator $q$ of this $D$ upwards. In each sequent with this occurrence of $q$, the succedent is not $q$. Indeed, for each application of $({\BS}\to)$ it either goes to the right premise (whose succedent keeps the same, and it is not $q$), or goes to the left premise with a formula of the form $p_{\widetilde{Z}_1} \wedge \ldots \wedge p_{\widetilde{Z}_k}$ as its succedent (each $p_{\widetilde{Z}_i}$ is not $p_{\widetilde{S}}$). Having $q$ as its succedent is impossible, since there is no other $D$, decomposed by $({\BS}\to)$, to the right of our $q$.

  On the other hand, the axiom should be of the form $q \to q$, with $q$ in the succedent. Contradiction.
\end{proof}

Now we are ready to state and prove the key lemma and prove the main theorem of this section.

\begin{lemma}\label{Lm:emptymain}
If $E, \Pi^D \to q$ is derivable, then either so is $\Pi \to p_{\widetilde{S}}$, or $\Pi$ is empty.
\end{lemma}

\begin{proof}
 Consider a cut-free derivation of $E, \Pi^D \to q$ and let us locate the principal occurrence of $q$. Consider two cases.

 {\em Case 1:} the principal occurrence of $q$ is in $E$. Then $\Pi$ is empty by Lemma~\ref{Lm:principal}.

 {\em Case 2:} the principal occurrence of $q$ is in $\Pi^D$. Let the formula with this principal occurrence be  $A = (p_{\widetilde{Y}_1} \wedge \ldots \wedge p_{\widetilde{Y}_k}) \BS (E \BS q) \SL (p_{\widetilde{Z}_1} \wedge \ldots \wedge \wedge p_{\widetilde{Z}_m})$, which is an element of a conjunction $B$ in $\Pi$. By Lemma~\ref{Lm:principal}, we have the following derivable sequents:
 \begin{align*}
  & E, \Pi_1^D \to E, \\
  & \Pi_2^D \to p_{\widetilde{Y}_1} \wedge \ldots \wedge p_{\widetilde{Y}_k}, \\
  & \Pi_3^D \to p_{\widetilde{Z}_1} \wedge \ldots \wedge p_{\widetilde{Z}_m}.
 \end{align*}
Here $\Pi_3$ is the part of $\Pi$ to the right of $A$, and $\Pi_1, \Pi_2$ is the one to the left. (Since ${}\to E$ is not derivable, the antecedent $\Psi_1$ for the sequent with succedent $E$ should include the leftmost formula $E$.)

By Lemma~\ref{Lm:emptyPi}, $\Pi_1$ is empty.
By Lemma~\ref{Lm:sidesequents}, the following sequents are derivable:
\begin{align*}
  & \Pi_2 \to p_{\widetilde{Y}_1} \wedge \ldots \wedge p_{\widetilde{Y}_k}, \\
  & \Pi_3 \to p_{\widetilde{Z}_1} \wedge \ldots  \wedge p_{\widetilde{Z}_m}.
\end{align*}

Now we have $\Pi = \Pi_2, B', \Pi_3$, where $B = B'[p_{\widetilde{S}} := D]$. In particular, the conjunction $B'$ includes the formula $A' = (p_{\widetilde{Y}_1} \wedge \ldots \wedge p_{\widetilde{Y}_k}) \BS p_{\widetilde{S}} \SL (p_{\widetilde{Z}_1} \wedge \ldots \wedge p_{\widetilde{Z}_m})$.

Now the goal sequent is derived as follows:
\[
\infer
{\Pi_2, B', \Pi_3 \to p_{\widetilde{S}}}
{\infer[({\wedge}\to)]{\ldots}
{\infer[({\SL}\to)]{\Pi_2, (p_{\widetilde{Y}_1} \wedge \ldots \wedge p_{\widetilde{Y}_k}) \BS p_{\widetilde{S}} \SL (p_{\widetilde{Z}_1} \wedge \ldots \wedge  p_{\widetilde{Z}_m}), \Pi_3 \to
p_{\widetilde{S}}
}{\Pi_3 \to p_{\widetilde{Z}_1} \wedge \ldots \wedge  p_{\widetilde{Z}_m} &
\infer[({\SL}\to)]{
\Pi_2, (p_{\widetilde{Y}_1} \wedge \ldots \wedge p_{\widetilde{Y}_k}) \BS p_{\widetilde{S}} \to
p_{\widetilde{S}}
}
{\Pi_2 \to p_{\widetilde{Y}_1} \wedge \ldots \wedge p_{\widetilde{Y}_k} &
p_{\widetilde{S}} \to p_{\widetilde{S}}}
}}}
\]
\end{proof}

\begin{proof}[Proof of Theorem~\ref{MALCs_empty_string}]
 Let $L$ be a language generated by a conjunctive grammar. If $\varepsilon \notin L$, then by Theorems~\ref{conjunctive_iff_conjunctive_categorial_theorem} and~\ref{conjunctive_categorial_to_MALC_transformation_theorem} we
 construct a $\MALC$-grammar $G$ for $L$. As noticed above, replacing $\MALC$ by $\MALCs$ does not change the language described by $G$. Hence, $G$ is the desired $\MALCs$-grammar.

 If $\varepsilon \notin L$, consider the language $L - \{ \varepsilon \}$ and apply Theorems~\ref{conjunctive_iff_conjunctive_categorial_theorem} and~\ref{conjunctive_categorial_to_MALC_transformation_theorem} to this language, yielding a $\MALC$-grammar $G$. Consider $G$ as a $\MALCs$-grammar and substitute $D$ for $p_{\widetilde{S}}$ everywhere in this grammar: $G^D = G[p_{\widetilde{S}} := D]$. Since the sequent ${} \to D$ is derivable, $\varepsilon$ belongs to the language described by $G'$.

 Now consider a non-empty string and let $\Pi$ be the sequence of formulae, assigned to this string by the original grammar $G$. The corresponding sequence for the new grammar~$G'$ is $\Pi^D$. If $\Pi \to p_{\widetilde{S}}$ were derivable, then so is $\Pi^D \to D$. Therefore, $L - \{ \varepsilon \}$ is a subset of the language described by $G^D$. On the other hand, if $\Pi^D \to D$ were derivable, then, by inverting $(\to{\BS})$, so is
 $E, \Pi^D \to q$. By Lemma~\ref{Lm:emptymain}, we get derivability of $\Pi \to p_{\widetilde{S}}$ (recall that $\Pi$ is non-empty). Hence, any non-empty string of the language described by $G'$ belongs to
 $L - \{ \varepsilon \}$.

 Thus, $G^D$ describes exactly the language $(L - \{ \varepsilon \}) \cup \{ \varepsilon \} = L$.
\end{proof}


\section{Lambek Grammars with Disjunction}\label{S:disjunction}

The translation from conjunctive categorial grammars to grammars based on $\MALC$ (i.e., Lambek grammars with additives) actually uses only one additive operation, namely conjunction. This is unsurprising, since this operation exactly corresponds to the conjunction operation used in conjunctive categorial grammars.

In this section, however, we show that a dual result also holds. Namely, for any conjunctive categorial grammar there exists a Lambek grammar with additive disjunction (but without additive conjunction) which describes the same language.

The idea is based on the following equivalence:
\[
 (A_1 \BS B) \wedge \ldots \wedge (A_n \BS B) \equiv
 (A_1 \vee \ldots \vee A_n) \BS B,
\]
which allows replacing conjunction by disjunction. As usual, $E \equiv F$ means simultaneous derivability of $E \to F$ and $F \to E$. If $E \equiv F$, then replacing $E$ with $F$ in a sequent derivable in $\MALC$ keeps its derivability, i.e., $\MALC$ enjoys the property of equivalent replacement.

Unfortunately, in the conjunctions we want to simulate using disjunction, the formulae are not of the form $A_i \BS B$ with the same $B$. In order to fit the formulae into the given shape, we shall use the ``relative double-negation'' construction (see Buszkowski~\cite{Buszkowski2007}, Kanovich et al.~\cite{Kanovich2019WoLLIC,Kanovich2022IC}):
\[
 A^{ff} = (A \BS f) \BS f,
\]
where $f$ is a variable. (In a ``real'' intuitionistic-style negation, $f$ should be the falsity constant; however, in $\MALC$ we have no such constant.)

The crucial property of this construction is given by the following lemma.
\begin{oldlemma}[{Kanovich et al.~\cite[Lem.~2]{Kanovich2019WoLLIC}}]\label{Lm:doublenegation}
 If $f$ does not occur in $A_1, \ldots, A_n \to B$, then this sequent is equiderivable with $A_1^{ff}, \ldots, A_n^{ff} \to B^{ff}$.
\end{oldlemma}

As in Theorem~\ref{conjunctive_categorial_to_MALC_transformation_theorem} above, let the given conjunctive categorial grammar have the following axiomatic propositions:
$$\begin{matrix}
A_{1,1} (a_1), & A_{1,2}(a_1), & \ldots & A_{1,k_1}(a_1),\\
A_{2,1} (a_2), & A_{2,2}(a_2), & \ldots & A_{2,k_2}(a_2),\\
\vdots\\
A_{n,1} (a_n), & A_{n,2}(a_n), & \ldots & A_{n,k_n}(a_n),
\end{matrix}$$
and let $s$ be its target category.
Let $f$ be a fresh variable, not occurring in the formulae $A_{i,j}$, and $f \ne s$. This variable will be used in the relative double negation constructions.

Let us replace each variable occurrence $p$ in each $A_{i,j}$ with its relative double negation $p^{ff}$, and let us denote the resulting formula by $A'_{i,j}$.  Since $A_{i,j} \in \Cat$, it includes the conjunction operation only in conjuncts of the form $p_1 \wedge \ldots \wedge p_k$. After the aforementioned substitution, such a conjunct becomes $p^{ff}_1 \wedge \ldots \wedge p^{ff}_k$, which is equivalent (in $\MALC$) to the formula $((p_1 \BS f) \vee \ldots \vee (p_k \BS f)) \BS f$. The latter uses disjunction instead of conjunction. Let us apply the equivalent replacement of each $p^{ff}_1 \wedge \ldots \wedge p^{ff}_k$ with $((p_1 \BS f) \vee \ldots \vee (p_k \BS f)) \BS f$ in $A'_{i,j}$ and obtain the formula denoted by $A''_{i,j}$.

Next, define $B'_i = (A'_{i,1} \wedge \ldots \wedge A'_{i,k_i})^{ff}$, for each $i$. Again, $B'_i$ is equivalent to a formula with only disjunction, not conjunction: $B'_i \equiv B''_i =
((A''_{i,1} \BS f) \vee \ldots \vee
(A''_{i,k_i} \BS f)) \BS f$.

This construction yields the desired equivalence theorem.

\begin{theorem}\label{conjunctive_categorial_to_disjunction}
The Lambek grammar with atomic propositions $B''_i(a_i)$, for $i = 1, \ldots, n$, and the goal category $(s^{ff})^{ff}$ describes the same language as the original conjunctive categorial grammar. Moreover, categories of this Lambek grammar use only $\BS$, $\SL$, and $\vee$, but not $\wedge$.
\end{theorem}

Before proving this theorem, let us recall the definition of another logical system, namely, {\em multiplicative-additive cyclic linear logic,} denoted by $\MACLL$. This logic is a non-commutative variant of Girard's linear logic (see Girard~\cite{Girard1987}, Yetter~\cite{Yetter}). Moreover, $\MACLL$ does not include exponential operations of linear logic, because we shall not need them in our arguments. We formulate $\MACLL$, following Kanovich et al.~\cite{Kanovich2019}, as a one-sided sequent calculus with tight negations.

{\em Atoms} are variables ($p,q,r,\ldots$, corresponding to primitive categories) and their negations: $\bar{p}, \bar{q}, \bar{r}, \ldots$ Formulae of $\MACLL$ are built from atoms and constants using operations, taken from the following table:

\vskip 3pt
\begin{center}
\begin{tabular}{|l|c|c|c|c|} \hline
& \multicolumn{2}{c|}{binary operations} & \multicolumn{2}{c|}{constants} \\ \hline
& conjunction & disjunction & true & false \\ \hline
multiplicative & $\otimes$ & $\Par$ & $1$ & $\bot$ \\\hline
additive & $\&$ & $\oplus$ & $\top$ & $0$ \\\hline
 \end{tabular}
 \end{center}

\vskip 3pt
Negation of an arbirary formula $A$, denoted by $A^\bot$, is defined externally as follows:
\begin{align*}
& p^\bot = \bar{p} \mbox{ ($p \in \PC$)} \\
& \bar{p}^\bot = p \mbox{ ($p \in \PC$)} \\
& (A \otimes B)^\bot = B^\bot \Par A^\bot && 1^\bot = \bot\\
& (A \Par B)^\bot = B^\bot \otimes A^\bot && \bot^\bot = 1\\
& (A \oplus B)^\bot = A^\bot \mathop{\&} B^\bot && 0^\bot = \top \\
& (A \mathop{\&} B)^\bot = A^\bot \oplus B^\bot && \top^\bot = 0
\end{align*}

Sequents of $\MACLL$ are expressions of the form ${} \to \Gamma$, where $\Gamma = A_1, \ldots, A_n$ is a non-empty sequence of $\MACLL$ formulae.

The notion of substitution of formulae for variables in $\MACLL$, in the presence of tight negations, is a bit more involved. Namely, the operation $[p := A]$ replaces each occurrence of atom~$p$ with $A$ and each occurrence of atom~$\bar{p}$ with $A^\bot$.

Axioms and inference rules of $\MACLL$ are depicted in Figure~\ref{Fig:MACLL}. Notice that there is no rule for constant $0$ (additive falsity); one may only introduce it together with $\top$, using $(\mathrm{ax})$ or $(\top)$.

\begin{figure*}
 $$
{} \to A, A^\bot
$$

$$
\infer[(\otimes)]{{} \to \Gamma, A \otimes B, \Delta}{{}\to \Gamma, A & {} \to B, \Delta}
\qquad
\infer[(\Par)]{{} \to A \Par B, \Gamma}{{} \to A, B, \Gamma}
$$

$$
\infer[(\mathop{\&})]{{}\to A_1 \mathop{\&} A_2, \Gamma}{{} \to A_1, \Gamma & {} \to A_2, \Gamma}
\qquad
\infer[(\oplus)_1]{{}\to A_1 \oplus A_2, \Gamma}{{} \to A_1, \Gamma}
\qquad
\infer[(\oplus)_2]{{}\to A_1 \oplus A_2, \Gamma}{{} \to A_2, \Gamma}
$$

$$
\infer[(1)]{\hspace*{.5em} {}\to 1\hspace*{.5em}}{}
\qquad
\infer[(\bot)]{{}\to \bot, \Gamma}{{}\to \Gamma}
\qquad
\infer[(\top)]{{}\to \top, \Gamma}{}
$$

$$
\infer[(\mathrm{cycle})]{{}\to \Gamma, \Delta}{{}\to \Delta,\Gamma}
$$
 \caption{Axioms and Rules of $\MACLL$}\label{Fig:MACLL}
\end{figure*}

Two formulae $A$ and $B$ are called equivalent (notation: $A \equiv B$), if the sequents $\to A^\bot, B$ and $\to B^\bot, A$ are both derivable. Interchanging equivalent formulae does not alter derivability.

The following translation maps Lambek formulae with additives to $\MACLL$ formulae:
\begin{align*}
& \widehat{p} = p\mbox{ ($p \in \PC$)} && 
 \widehat{A \cdot B} = \widehat{A} \otimes \widehat{B} \\ 
& \widehat{A \BS B} = \widehat{A}^\bot \Par \widehat{B} && \widehat{A \wedge B} = \widehat{A} \mathop{\&} \widehat{B}\\
&  \widehat{B \SL A} = \widehat{B} \Par \widehat{A}^\bot && \widehat{A \vee B} = \widehat{A} \oplus \widehat{B}
\end{align*}
This translation actually gives an exact (conservative) embedding of $\MALCs$ into $\MACLL$.

\begin{oldtheorem}[{Kanovich et al.~\cite[Cor.~4]{Kanovich2019}}]\label{MALC-into-MACLL}
 The sequent $A_1, \ldots, A_n \to B$ is derivable in $\MALCs$ if and only if its translation ${} \to \widehat{A}_n^\bot, \ldots, \widehat{A}_1^\bot, \widehat{B}$ is derivable in $\MACLL$.
\end{oldtheorem}

This theorem is a  non-commutative variant of an embedding result by Schellinx~\cite{Schellinx1991} for intuitionistic and classical versions of commutative linear logic. Pentus~\cite{Pentus1998} proved a weaker version of Theorem~\ref{MALC-into-MACLL}, without additives.



\begin{proof}[Proof of Theorem~\ref{conjunctive_categorial_to_disjunction}]
Let us first prove Theorem~\ref{conjunctive_categorial_to_disjunction} for $\MALCs$.
 As in Theorem~\ref{conjunctive_categorial_to_MALC_transformation_theorem}, we establish the following statement. For any string $a_{i_1} \ldots a_{i_m}$ the statement $s(a_{i_1} \ldots a_{i_m})$ is derivable in the conjunctive categorial grammar if and only if the sequent $B''_{i_1}, \ldots, B''_{i_m} \to (s^{ff})^{ff}$ is derivable in $\MALCs$. 

 For convenience, let us replace each $B''_i$ with its equivalent $B'_i$. These formulae contain conjunction, but this is used only inside the proof, not in the final statement. Due to cut elimination, the goal sequents of the form $B''_{i_1}, \ldots, B''_{i_m} \to (s^{ff})^{ff}$ will still have conjunction-free derivations.

 Let us start with {\bf the ``if'' part.} Let the sequent $B'_{i_1}, \ldots, B'_{i_m} \to (s^{ff})^{ff}$ be derivable in $\MALCs$. Consider its translation into $\MACLL$, ${}\to (\widehat{B}'_{i_m})^\bot, \ldots,
 (\widehat{B}'_{i_1})^\bot, \widehat{(s^{ff})^{ff}}$. This sequent is derivable by Theorem~\ref{MALC-into-MACLL}.

 Now let us substitute the multiplicative falsity constant $\bot$ for $f$. The sequent keeps being derivable.
 It is easy to show that after such substitution each formula of the form $\widehat{F^{ff}}$ becomes equivalent to $\widehat{F}$. Indeed,
 $\widehat{F^{ff}} =  (\bar{f} \otimes \widehat{F}) \Par f$ transforms into $(1 \otimes \widehat{F}) \Par \bot \equiv \widehat{F}$ (since $1$ is neutral for $\otimes$ and $\bot$ is neutral for $\Par$). In other words, substituting $\bot$ for $f$ transforms relative double negation into ``real'' double negation, which can be removed due to the classical nature of $\MACLL$.

 Thus, after this substitution we get derivability of the sequent $\to (\widehat{B}_{i_m})^\bot, \ldots,\linebreak (\widehat{B}_{i_1})^\bot, s$ in $\MACLL$. By Theorem~\ref{MALC-into-MACLL}, its preimage $B_{i_1}, \ldots, B_{i_m} \to s$ is derivable in $\MALCs$. As noticed above, for sequents of this form derivability in $\MALCs$ is equivalent to that in $\MALC$. Finally, the ``if'' part of Theorem~\ref{conjunctive_categorial_to_MALC_transformation_theorem} gives derivability of $s(a_{i_1} \ldots a_{i_m})$ in the conjunctive categorial grammar.

 {\bf The ``only if'' part} is proven using the relative double negation lemma. Let $s(a_{i_1} \ldots a_{i_m})$ be derivable in the conjunctive categorial grammar. By the ``only if'' part of Theorem~\ref{conjunctive_categorial_to_MALC_transformation_theorem}, the sequent $B_{i_1}, \ldots, B_{i_m} \to s$ is derivable in $\MALC$, and therefore in $\MALCs$. Let us substitute $p_j^{ff}$ for each variable $p_j$ (in particular, $s^{ff}$ for $s$). Since $B_i = A_{i_1} \wedge \ldots \wedge A_{i,k_i}$, this yields derivability of the following sequent:
 \[
  (A'_{i_1,1} \wedge \ldots \wedge A'_{i_1,k_{i_1}}), \ldots,
  (A'_{i_m,1} \wedge \ldots \wedge A'_{i_m,k_{i_m}}) \to s^{ff}.
 \]
 Now Lemma~\ref{Lm:doublenegation} yields derivability of its relative double negation translation:
 \[
  (A'_{i_1,1} \wedge \ldots \wedge A'_{i_1,k_{i_1}})^{ff}, \ldots,
  (A'_{i_m,1} \wedge \ldots \wedge A'_{i_m,k_{i_m}})^{ff} \to (s^{ff})^{ff},
 \]
or $B'_{i_1}, \ldots, B'_{i_m} \to (s^{ff})^{ff},$ which is equivalent the desired sequent:
 \[
  B''_{i_1}, \ldots, B''_{i_m} \to (s^{ff})^{ff}.
 \]
%
%
%
%
%
%
 
 In order to cover the case of $\MALC$, we claim that for sequents of the form $B''_{i_1}, \ldots, B''_{i_m} \to (s^{ff})^{ff}$ (which are used in the grammar
constructed above) their derivability in $\MALC$ is equivalent to the one in $\MALCs$. Indeed, for any subformula $F$ of such a sequent, ${} \to F$ is not derivable.  Hence, each $\MALCs$ derivation actually obeys Lambek's non-emptiness restriction, i.e., it is a $\MALC$ derivation.
\end{proof}

Languages with the empty string can also be handled in this setting, yielding the following variant of Theorem~\ref{MALCs_empty_string}.

\begin{theorem}\label{conjunctive_categorial_to_disjunction_empty}
Any language described by a conjunctive grammar can be generated by a $\MALCs$-grammar whose categories use only $\BS$, $\SL$, and $\vee$.
\end{theorem}

\begin{proof}
Let us first consider the case where the given language $L$ does not include $\varepsilon$.  If $\varepsilon \notin L$, one just takes the grammar constructed in Theorem~\ref{conjunctive_categorial_to_disjunction} as the desired $\MALCs$-grammar.

Now let $\varepsilon \in L$. Then let us substitute, in the given $\MALCs$-grammar with disjunction, $D = \bigl( (r\BS r) \BS ((t \BS t) \BS q) \bigr) \BS q$ for $s$. This substitution adds $\varepsilon$ to the language and keeps there all the non-empty strings from $L$. It remains to show that no extra non-empty string becomes generated by the grammar. This is proved by the same argument as in the ``if'' part of Theorem~\ref{conjunctive_categorial_to_disjunction}. Namely, we replace the sequent $B''_{i_1}[s := D], \ldots, B''_{i_m}[s := D] \to (D^{ff})^{ff}$ with the equivalent sequent
$B'_{i_1}[s := D], \ldots, B'_{i_m}[s := D] \to (D^{ff})^{ff}$ and translate it into $\MACLL$:
\[
{} \to (\widehat{B}'_{i_m})^\bot[s := \widehat{D}], \ldots, (\widehat{B}'_{i_1})^\bot[s := \widehat{D}], \widehat{(D^{ff})^{ff}}.
\]
Next, we substitute $\bot$ for $f$ and get derivability of ${} \to (\widehat{B}_{i_m})^\bot[s := \widehat{D}], \ldots, (\widehat{B}_{i_1})^\bot[s := \widehat{D}], D$ in $\MACLL$, and therefore
derivability of $B_{i_1}[s := D], \ldots, B_{i_m}[s := D] \to D$ in $\MALCs$. By construction from Section~\ref{S:empty} (Lemma~\ref{Lm:emptymain}), this entails, for $m > 0$, derivability of $B_{i_1}, \ldots, B_{i_m} \to s$ in $\MALCs$. Hence, the given non-empty string belongs to the original language $L$.
\end{proof}

Finally, as a corollary, we get the following version of Theorem~I by Kanazawa~\cite{Kanazawa1992}.

\begin{corollary}
Any finite intersection of context-free languages is described by a $\MALC$-grammar using only $\BS$, $\SL$, and $\vee$. The same holds for $\MALC$, provided the given intersection does not include $\varepsilon$.
\end{corollary}



\section{$\MALC$ with Distributivity}\label{S:distributivity}

In this section, we shall discuss issues connected with distributivity. It is well known that the distributivity law for $\wedge$ and $\vee$, namely,
\[
(A \vee B) \wedge C \to (A \wedge C) \vee (B \wedge C),
\]
is not derivable in substructural logics like the $\MALC$ or $\MALCs$ (see Ono and Komori~\cite{OnoKomori}). On the other hand, natural interpretations of the Lambek calculus
on formal languages suggest this principle to be valid. This motivates considering extensions of $\MALC$ and $\MALCs$ with the distributivity principle as an extra axiom. We denote
such extensions by $\MALCD$ and $\MALCDs$ respectively. Cut-free calculi for such systems were given by Kozak~\cite{Kozak}.

We shall show that distributivity does not affect our results on embedding conjunctive grammars into multiplicative-additive Lambek grammars. The easier part here is the conjunction-only one.
Namely, the fragments of $\MALC$ and $\MALCD$ in the language of $\BS$, $\SL$, and $\wedge$ simply coincide, and the same for $\MALCs$ and $\MALCDs$ (see Kanovich et al.~\cite[Thm.~6]{Kanovich2022IC}). In other words,
in this restricted language adding the distributivity principle does not extend the set of derivable sequents. This gives the following corollaries.
\begin{theorem}
\begin{enumerate}
\item If a language without the empty string is described by a conjunctive grammar, then it can also be described by a $\MALCD$-grammar using only $\BS$, $\SL$, $\wedge$.
\item If a language is described by a conjunctive grammar, then it can also be described by a $\MALCDs$-grammar using only $\BS$, $\SL$, $\wedge$.
\end{enumerate}
\end{theorem}

The situation with disjunction is a bit trickier. Namely, there exists a sequent using only $\BS$, $\SL$, and $\vee$, which is derivable in $\MALCD$, but not in $\MALC$ or even $\MALCs$ (Kanovich et al.~\cite[Thm.~7]{Kanovich2022IC}). Nevertheless, the results for conjunctive grammars still hold.
\begin{theorem}
\begin{enumerate}
\item If a language without the empty string is described by a conjunctive grammar, then it can also be described by a $\MALCD$-grammar using only $\BS$, $\SL$, $\vee$.
\item If a language is described by a conjunctive grammar, then it can also be described by a $\MALCDs$-grammar using only $\BS$, $\SL$, $\vee$.
\end{enumerate}
\end{theorem}

\begin{proof}
The sequents used in grammars provided by Theorems~\ref{conjunctive_categorial_to_disjunction} and~\ref{conjunctive_categorial_to_disjunction_empty} use $\BS$, $\SL$, and $\vee$, but can be equivalently rewritten in the language of $\BS$, $\SL$, $\wedge$. Namely, one replaces $B''_i$ by $B'_i$ and $B''_i[s := D]$ with $B'_i[s := D]$. For such sequents, derivability in $\MALCD$ is equivalent to the one in $\MALC$, and derivability in $\MALCDs$ is equivalent to the one in $\MALCs$. Next, since the equivalences $B''_i \equiv B'_i$ and $B''_i[s := D] \equiv B'_i[s := D]$ keep valid in the presence of distributivity, we conclude that derivability sequents used for grammars provided by  Theorems~\ref{conjunctive_categorial_to_disjunction} and~\ref{conjunctive_categorial_to_disjunction_empty} does not depend on distributivity. Therefore, if one considers the same grammars as grammars over $\MALCD$ or $\MALCDs$ respectively, they will describe the same languages.
\end{proof}

It is also worth mentioning that for the {\em non-associative} version of $\MALCD$ the situation is significantly different, and the expressive power is much smaller. Namely, as shown by Buszkowski and
Farulewski~\cite{BuszkowskiFarulewski}, grammars based on the distributive non-associative multiplicative-additive Lambek calculus describe exactly the class of context-free languages.


\section{Conclusion}

We have developed a new family of categorial grammars, called {\em conjunctive categorial grammars,}  which extend basic categorial grammars with conjunction. We have shown that these grammars are equivalent
 to conjunctive grammars, proposed by Okhotin, provided the languages do not include the empty string. Conjunctive categorial grammars may be embedded into Lambek grammars extended with additive operations
 (this formalism was proposed by Kanazawa), and even one of the two additives---conjunction or disjunction---is sufficient for such an embedding. We have shown, however, that there exists a Lambek grammar with additive conjunction
 ($\MALC$-grammar)  which  describes an NP-hard language. Therefore, under the condition that $\mathrm{P} \ne \mathrm{NP}$, Kanazawa's framework is strictly stronger than conjunctive grammars.

 Several questions are left open for further research.
  First, we do not yet have an unconditional separation of Lambek grammars with additives and conjunctive grammars. Second, it would be interesting to establish natural bounds on the class of languages described by $\MALC$-grammars.
  In particular, it is unknown whether such grammars can describe languages which are not images of conjunctive languages under symbol-to-symbol homomorphisms. Third, there is an open question whether there exists
  a $\MALC$-grammar describing a language harder than NP-complete (e.g., a PSPACE-complete one).


\subsection*{Acknowledgements}

The part of this article due to Stepan Kuznetsov was prepared within the framework of the HSE University Basic Research Program. The work of Stepan Kuznetsov was supported by the Theoretical Physics and Mathematics Advancement Foundation ``BASIS.''


\end{document}